\documentclass[11pt, reqno]{amsart}
\usepackage{amsfonts,latexsym, enumerate}
\usepackage{amsmath}
\usepackage{amscd}
\usepackage{float,amsmath,amssymb,mathrsfs,bm,multirow,graphics}
\usepackage[dvips]{graphicx}
\usepackage[percent]{overpic}

\newcommand{\nequation}{\setcounter{equation}{0}}
\renewcommand{\theequation}{\mbox{\arabic{section}.\arabic{equation}}}
\newcommand{\R}{{\Bbb R}}

\newcommand{\C}{{\Bbb C}}

\newcommand{\Z}{{\Bbb Z}}
\newcommand{\proofbegin}{\noindent{\it Proof.\,\,}}
\newcommand{\proofend}{\hfill$\Box$\bigskip}

\DeclareMathOperator{\im}{Im}
\DeclareMathOperator{\re}{Re}

\newcommand{\res}{\text{\upshape Res\,}}
\newcommand{\dist}{\text{\upshape dist}}

\newtheorem{theorem}{Theorem}[section]
\newtheorem{proposition}[theorem]{Proposition}
\newtheorem{corollary}[theorem]{Corollary}
\newtheorem{lemma}[theorem]{Lemma}

\newtheorem{remark}[theorem]{Remark}
\newtheorem{example}[theorem]{Example}
\newtheorem{figuretext}{Figure}

\input epsf
\title[The KdV equation on the half-line]{The KdV equation on the half-line: \\ Time-periodicity and mass transport}
\author{Jerry L. Bona}
\address{J.L.B.: Department of Mathematics, Statistics and Computer Science, University of Illinois at Chicago, 851 S. Morgan Street MC 249, 
Chicago, Il 60607, USA.}  \email{\tt jbona@uic.edu}

\author{Jonatan Lenells}
\address{J.L.: Department of Mathematics, KTH Royal Institute of Technology, Stockholm, Sweden}
\email{\tt jlenells@kth.se}

\begin{document}

\begin{abstract} 
\noindent  The work presented here emanates from questions arising from experimental observations of the propagation of surface water waves.  
The experiments in question featured a periodically moving wavemaker located at one end of a flume that generated unidirectional waves of
relatively small amplitude and long wavelength when compared with the undisturbed depth.  It was observed that the wave profile at any point
down the channel very quickly became periodic in time with the same period as that of the wavemaker.  One of the questions dealt with here 
is whether or not such a property holds for model equations for such waves.  In the present discussion, this is examined in the context of the 
Korteweg-de Vries equation using the recently developed version of the inverse scattering theory for boundary value problems put forward
by Fokas and his collaborators.  It turns out that the Korteweg-de Vries equation does possess the properly that solutions at a fixed point
down the channel  have the property of asymptotic periodicity in time when forced periodically at the boundary.
  However, a more subtle issue to do with conservation of mass 
fails to hold at the second order in a small parameter which is the typical  wave amplitude divided by the undisturbed depth.    
\end{abstract}

\maketitle

\noindent
{\small{\sc AMS Subject Classifications (2010)}:  35B30, 35C15, 35Q53, 37K10, 37K15, 76B15, 86-05.


%\tableofcontentss

\section{Introduction}\nequation
The propagation of long-crested, unidirectional, small-amplitude, long wavelength disturbances 
 over a featureless,  flat bottom in shallow water can be      approximately described 
by  Korteweg-de Vries--type  equations.  A one-parameter class of such 
equations takes the form
\begin{align}\label{kdvalphabeta}
  u_t + u_x + \alpha uu_x + \beta\big (\mu u_{xxx} - (1 - \mu)u_{xxt}\big) = 0.
\end{align}
Here, the independent variable $x$ is proportional to distance in the direction of propagation
while $t$ is proportional to elapsed time.  The dependent variable $u(x,t)$ is proportional to to the deviation of the free surface 
from its rest position at the point corresponding to $x$ at time $t$.  
The real parameters 
$$
\alpha = \frac{a}{h} \quad {\rm and}  \quad  \beta = \frac{h^2}{\lambda^2}
$$
 are defined in terms of a typical amplitude $a$ and wavelength $\lambda$ as well as the undisturbed depth $h$ of the water.  
 The variables are scaled so that $u$ and its partial derivatives are formally all of order one while $\alpha$ and $\beta$ are assumed to 
 be small compared to one.  The parameter $\mu$ is  
a modeling parameter that in principle can take any real value.  However, the initial-value problem for the model will not
be well posed unless $\mu \leq 1$.  

In a flat-bottomed, laboratory channel, Zabusky and Galvin \cite{zabusky} ran experiments showing qualitative agreement 
between measurements and the model's predictions for the case $\mu = 1$, the classical Korteweg-de Vries equation (KdV equation henceforth).    
Later work by Hammack and Segur \cite{hammack} continued this line of investigation and also found qualitative agreement.  The detailed accuracy of such models in
the case $\mu = 0$, the BBM equation,  was investigated in a series of wave tank experiments 
reported in \cite{BPS1981}.   In these experiments, a paddle-type wavemaker mounted at  one end of the tank 
was oscillated periodically and
 the resulting wave motion was monitored at several points down the channel.  More precisely, 
 four measurements of the wave motion were taken at points $x = 0 < x_1 < x_2 < x_3$.
 This produced four time series, $u(0,t), u(x_1,t), u(x_2,t)$ and $u(x_3,t), \, t \geq 0$.     
 
These measurements suggested an initial-boundary-value problem for \eqref{kdvalphabeta} wherein the measurement  
  $u(0,t) = g_0(t)$ was taken as boundary data for the equation and  the initial data $u(x,0) = u_0(x)$ was  
  identically  zero, corresponding to the water
 being initially at rest.  The predictions of this initial-boundary-value problem were then compared directly to the other three time series.  
 As the equation \eqref{kdvalphabeta}
 is an 
approximation for waves moving only to the right, the experiment ceases as soon as
the waves reach the end of the channel and reflection becomes relevant. 
 Hence, a natural boundary condition at $x = L$, the end of the tank, is $u(L,t) = 0$ and,
in the case of the KdV equation, the second boundary condition $u_x(L,t) = 0$ is also required.  
However, since a lateral boundary condition at the end of
the channel away from the wavemaker is irrelevant to the motion prior to the wave reaching it, it is 
mathematically easier to simply push the right-hand boundary to infinity.   Rigorous 
justification of this procedure on the time scale where there is no motion at $x = L$ can be found in \cite{BCSZ1} and \cite{BCSZ2}.   
The outcome of the comparisons made in  \cite{BPS1981} is that the just-described initial-boundary-value problem works
 quantitatively quite well, even for rather large 
values of the Stokes' number $S = \alpha/\beta$.  

In the course of examining the results of the experiments just described, two qualitative features of the wave motion emerged. 
 The goal of the present essay is to address  rigorously these two  aspects  in the context of the KdV equation, $\mu = 1$. 

\medskip\noindent
A. {\bf Time-periodicity.} 
The experimental data indicate that a periodically moving wavemaker gives rise to measurements $u(0,t)$ and $u(x_j,t), \, j = 1,2,3,$ which are asymptotically time-periodic with the same period as that of
the wavemaker. In other words, if the wave maker oscillates with period $t_p$, then the functions $u(0,t)$ and $u(x_j,t)$ approach functions which are periodic in $t$ with period $t_p$ as $t \to \infty$.  Indeed, in the experiments, this approach is seen to be very rapid.  

\medskip\noindent
B. {\bf Mass transport. }
In the experimental set-up, the total mass of the water in the channel is evidently constant. As the wavemaker oscillates periodically, the net amount of water added to the region beyond the first measurement point $x = 0$ oscillates accordingly. Thus the function $M(t) = \int_0^\infty u(x,t) dx$ is expected to settle down to periodic oscillations around zero for large values of $t$. 

\vspace{.15cm}

The plan of the paper is to introduce the principal theorems for both the linear problem in which the nonlinearity is dropped and the nonlinear problem in the next 
section.   
This will include a discussion of previous work on these problems.   Theorem 2.1 will be proved in 
Section 3 while the nonlinear Theorem 2.2 and the resulting Corollary 2.3 will be dealt with in Section 4. The proofs are inspired by the developments in 
\cite{FL2012} where the nonlinear Schr\"odinger equation with asymptotically time-periodic data was considered. 
As mentioned, attention will be given only to the  Korteweg-de Vries equation, the case $\mu = 1$. 
This is because the main tool used in the present investigation is inverse scattering theory. As far as we know,   
the model \eqref{kdvalphabeta} does not have an inverse scattering theory if $\mu \neq 1$.     A brief concluding section
includes not only a summary of the results, but brief remarks on the implications for wave tank experiments.

\section{Main results}\nequation\label{mainsec}
In this section,  the two principal results of our study are stated and discussed. The proofs 
are presented  in Sections \ref{linearsec} and \ref{nonlinearsec}.  An elementary rescaling of
the independent and dependent variables assures that we may take $\alpha = 6$ and $\beta = 1$ in (\ref{kdvalphabeta}). 
%$u(x,t)$ solves (\ref{kdvalphabeta}) with $\alpha, \beta$ if and only if $v(x,t) = \frac{\alpha}{6}u(\beta^\frac12 x, \beta^\frac12 t)$ solves (\ref{kdv}).
However, it must be remembered that the resulting boundary data  $\alpha g_0(\beta^\frac12 t)/6$ now depends upon the parameters $\alpha$ and $\beta$.    

The mathematical problem under consideration is then the equation
\begin{align}\label{kdv}
  u_t + u_x + 6uu_x + u_{xxx} = 0, \qquad x > 0, \quad t > 0,
\end{align}
together with the initial and boundary conditions
\begin{equation}  \label{auxiliary}
  u(x,0) = u_0(x) \equiv 0, \; {\rm for} \;  x \geq 0 \;\, {\rm and} \;\, u(0,t)  = g_0(t), \; {\rm for} \; t \geq 0,
\end{equation}
where the given Dirichlet boundary value $g_0$ is taken to be smooth and compatible with the vanishing initial data at $t = 0$, which is 
to say $g_0(0) = 0$. 
This initial-boundary value problem has received considerable attention.   It is known to be globally well posed in a variety of circumstances to do 
with restrictions on the initial and boundary data (these development started with    \cite{BW1} and  \cite{BW2}; see \cite{BSZ2008} and the references contained therein
for a more up-to-date appraisal).   As the present discussion derives directly from experimental results, we are not going to be specially concerned with 
sharp hypotheses on the data.   The answers to the issues just mentioned  are the focus.  Also, while the Korteweg-de Vries equation is known rigorously to approximate well
solutions of the full, inviscid water wave problem (see  \cite{AABCW}, \cite{BCL}, \cite{BPS1983},  \cite{craig}), that fact depends upon smoothness of the auxiliary data.  Without sufficient 
smoothness, there is no approximation.  

\subsection{The linear limit}
The first theorem answers the questions posed in the introduction in the affirmative in the case of the linearized version of equation (\ref{kdv}).

\begin{theorem}\label{linearth}
Let $u(x,t)$ be a sufficiently smooth solution of the linearized KdV equation 
\begin{align}\label{linearkdv}
  u_t + u_x + u_{xxx} = 0, \qquad x > 0, \quad t > 0.
\end{align}
with vanishing initial data and compatible, periodic Dirichlet data $u(0,t) = g_0(t)$ of period $t_p>0$, i.e.,
$$g_0(0) = 0 \quad \text{and}  \quad g_0(t + t_p) - g_0(t) = 0 \quad \text{for} \quad t > 0.$$ 
\begin{enumerate}[$(a)$]
\item For any fixed $x \geq 0$, $u(x,t)$ is asymptotically time-periodic with period $t_p$.  More precisely, for each $x \in \R$ and as $t \to \infty$, 
\begin{align}\label{linearqdifference}
  u(x,t + t_p) - u(x,t) = O(t^{-\frac{3}{2}}).
\end{align}

\item For any fixed $x \geq 0$, the mass function $M(x,t) = \int_x^\infty u(x,t) dx$ has the property that as $t \to \infty$, 
\begin{align}\label{linearMdifference}
  M(x, t + t_p) - M(x,t) = \int_0^{t_p} g_0(t')dt' + O(t^{-\frac{3}{2}}).
\end{align}
In particular, if $g_0(t)$ has zero average, then $M(x, t)$ is asymptotically time-periodic with period $t_p$.
\end{enumerate}
\end{theorem}

Asymptotic periodicity is established for the linear problem for both the KdV and the BBM equation in \cite{BW}.  These results 
are obtained  by classical energy 
estimates.   The theory reported there is not as sharp as that obtained here using inverse scattering techniques.   In any case,
 the linear inverse scattering theory is needed for our analysis of the nonlinear problem.

\subsection{The nonlinear problem}
In the  second theorem,   nonlinear corrections to  Theorem \ref{linearth} are kept and estimated. 
Thus, consider a perturbative solution
$$u(x,t) = \epsilon u_1(x,t) + \epsilon^2 u_2(x,t) + O(\epsilon^3)$$
of (\ref{kdv}) with Dirichlet data
$$g_0(t) = \epsilon g_{01}(t) + \epsilon^2 g_{02}(t)  + O(\epsilon^3),$$
where $\epsilon > 0$ is a small parameter.
The first and second Neumann boundary values of the solution are 
$$g_1(t) = u_x(0,t) \quad {\rm and} \quad g_2(t) = u_{xx}(0,t).$$
Their respective perturbative expansions are written as 
\begin{align}\label{g1g2}
g_1(t) = \epsilon g_{11}(t)  + \epsilon^2 g_{12}(t)  + O(\epsilon^3), \qquad
g_2(t) = \epsilon g_{21}(t) + \epsilon^2  g_{22}(t) + O(\epsilon^3).
\end{align}
For definiteness, the results are presented when the Dirichlet data comprise  a periodic sine-wave.   
Similar results can be obtained for other periodic boundary forcings. It is worth mentioning that
the measured boundary conditions in the experiments reported in \cite{BPS1981} closely resemble 
 a sine wave while the boundary conditions used in the  sediment transport study \cite{BBR}
 were modeled exactly as sine waves with the field measured amplitudes and frequencies.  

\begin{theorem}\label{mainth}
Let $\omega \in \R$ be a non-zero constant such that
\begin{align}\label{omegaassumption}
  |\omega| \neq \frac{2}{3\sqrt{3}} \quad \text{and}\quad |\omega| \neq \frac{1}{3\sqrt{3}}.
\end{align}
Let $u(x,t)$ be a solution of (\ref{kdv}) with boundary data
$$g_0(t) = \epsilon \sin{\omega t}.$$
Let $K,L \in \C$ denote the unique solutions of the cubic equations
$$4K^3 - K + \frac{\omega}{2} = 0, \qquad 4L^3 - L + \omega = 0,$$
such that $K$ and $L$ belong to the boundary $\partial D_3$ of the domain $D_3 \subset \C$ defined by
$$D_3 = \{\im k < 0\} \cap \{\im(4k^3 - k) > 0\}$$
(see Figure 1).  Then, the first and second Neumann boundary values of $u$ at $x=0$ are as in (\ref{g1g2}), where, as $t \to \infty$, 
\begin{subequations}  \label{gformulas}
\begin{align}  \label{g11end}
g_{11}(t) = &  - K  e^{i\omega t} 
 -\bar{K}  e^{-i\omega t} + O\big(t^{-\frac{3}{2}}\big), 
 	\\\label{g21end}
g_{21}(t) = & \; 2iK^2 e^{i\omega t}- 2i \bar{K}^2 e^{-i\omega t} + O\big(t^{-\frac{3}{2}}\big), 
	\\ \label{g12end}
g_{12}(t) = &\; \frac{1}{8i}\bigg(\frac{L}{K^2} - \frac{2}{K}\bigg)e^{2i\omega t}
+\frac{1}{2} \im \bigg(\frac{1}{K}\bigg)
  - \frac{1}{8i} \bigg(\frac{ \bar{L}}{\bar{K}^2}- \frac{2}{\bar{K}}\bigg)e^{-2i\omega t} 	
 + O\big(t^{-\frac{3}{2}}\big),
 	\\   \label{g22end}
g_{22}(t) = & \;  \bigg(1 - \frac{L^2}{4K^2}\bigg) e^{2i\omega t} 
+ \re\bigg(\frac{K}{\bar{K}}\bigg) -1  
+ \bigg(1 - \frac{\bar{L}^2}{4\bar{K}^2}\bigg) e^{-2i\omega t}  + O\big(t^{-\frac{3}{2}}\big). 
\end{align}
\end{subequations}
\end{theorem}

Theorem \ref{mainth} reveals that, at least to second order in perturbation theory, the time-periodic Dirichlet profile $\epsilon \sin{\omega t}$ 
gives rise to time-periodic Neumann conditions. This suggests that the time-periodic Dirichlet data also generates a time-periodic solution in the nonlinear case. 

Results of this nature for the nonlinear problem, but with damping incorporated, are discussed in \cite{BSZ2002}.   That theory relies upon rather delicate 
Fourier analysis and is not as sharp as what is brought forth here.  

The situation for  mass transport is more complicated.

\begin{corollary}
Under the assumptions of Theorem \ref{mainth}, the mass function $M(t) = \int_0^\infty u(x,t) dx$ satisfies
\begin{align}\label{Mdifference}
  M(t + t_p) - M(t) = \epsilon m_1(t) + \epsilon^2 m_2(t) + O(\epsilon^3),
\end{align}
where $t_p = 2\pi /\omega$ and, as $t \to \infty$, 
\begin{subequations}\label{m1m2}
\begin{align}
& m_1(t) = O\big(t^{-\frac{3}{2}}\big),
	\\ \label{m1m2b}
& m_2(t) =  \frac{\pi}{\omega}\bigg[1 + 2\re\bigg(\frac{K}{\bar{K}}\bigg)\bigg] 
+ O\big(t^{-\frac{3}{2}}\big).
\end{align}
\end{subequations}
\end{corollary}
\proofbegin
Calculate as follows:
\begin{align*}
M(t + t_p) - M(t) 
= & \int_0^\infty \int_t^{t+ t_p} u_t(x,s) ds dx
	\\
=& - \int_t^{t+ t_p}  \int_0^\infty  \big[u_x + 6 uu_x + u_{xxx}\big](x,s) dx ds
	\\
=& \int_t^{t+t_p}  \big[g_0(s) + 3 g_0^2(s) + g_2(s)\big] ds.
\end{align*}
It thus transpires that 
\begin{align*}
& m_1(t) = \int_t^{t+t_p}  \big[g_{01}(s) + g_{21}(s)\big] ds \quad {\rm and} 
	\\
& m_2(t) = \int_t^{t+t_p}  \big[3 g_{01}^2(s) + g_{22}(s)\big] ds.	
\end{align*}
Since $g_{01}(t) = \sin{\omega t}$, the expressions for $g_{21}$ and $g_{22}$ obtained in Theorem \ref{mainth} yield (\ref{m1m2}). 
\proofend

Equation (\ref{m1m2b}) implies that $m_2(t)$ does not approach zero as $t \to \infty$. Indeed, $\lim_{t \to \infty} m_2(t) = 0$ if and only if
$$\arg{K} = \pm \frac{\pi}{3} + \pi n \qquad {\rm for \; some} \;n \in \Z,$$
and this equation is never satisfies for $K \in \partial D_3$ and $K \neq 0$. 
This suggests that a periodic Dirichlet boundary condition does {\it not} in general give rise to an asymptotically periodic mass function $M(t) = \int_0^\infty u(x,t) dx$, 
although the discrepancy lies at second order.

\section{Proof of Theorem \ref{linearth}}\nequation\label{linearsec}
The first step is to derive an integral representation for the solution $u(x,t)$ of the  boundary value problem for the linear equation \eqref{linearkdv}.  

Define the open subsets $\{D_j\}_1^4$ of the complex $k$-plane by 
\begin{align*}
D_1 = \{\im k > 0\} \cap \{\im(4k^3 - k) > 0\},  \qquad
D_2 = \{\im k > 0\} \cap \{\im(4k^3 - k) < 0\}, 
	\\
D_3 = \{\im k < 0\} \cap \{\im(4k^3 - k) > 0\},  \qquad
D_4 = \{\im k < 0\} \cap \{\im(4k^3 - k) < 0\}.
\end{align*}	
Let $D_1 = D_1' \cup D_1''$ where $D_1' = D_1 \cap \{\text{Re}\, k > 0\}$ and $D_1'' = D_1 \cap \{\text{Re}\, k < 0\}$. Similarly, let $D_4 = D_4' \cup D_4''$ with $D_4' = D_4 \cap \{\text{Re}\, k > 0\}$ and $D_4'' = D_4 \cap \{\text{Re}\, k < 0\}$; see again Figure \ref{DjsKdV.pdf}.
\begin{figure}
\begin{center}
\begin{overpic}[width=.5\textwidth]{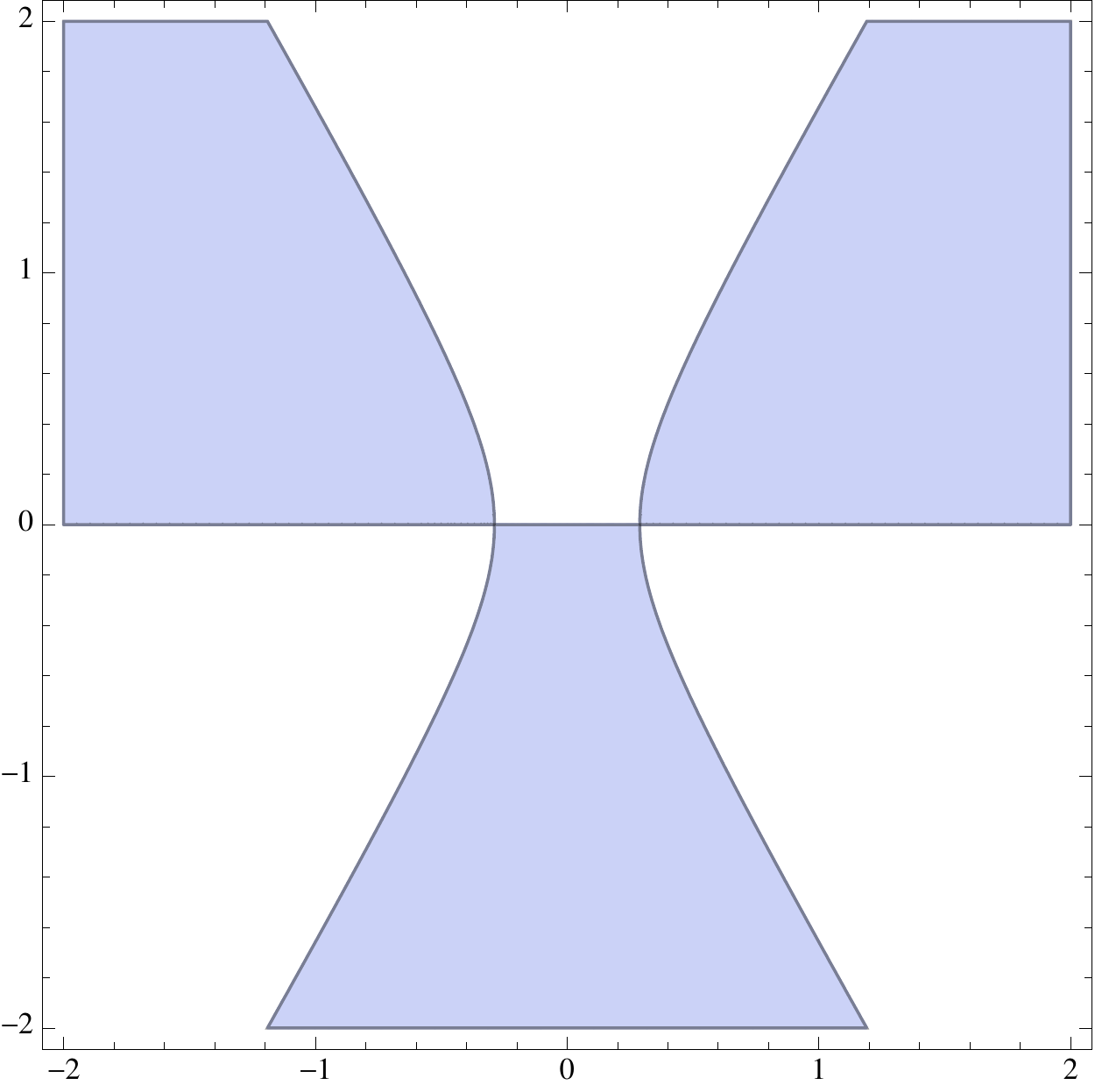}
      \put(77,67){$D_1'$}
      \put(48,75){$D_2$}
      \put(17,67){$D_1''$}
      \put(17,30){$D_4''$}
      \put(48,22){$D_3$}
      \put(77,30){$D_4'$}
%       \put(105,50){$\text{Re}\; k$}
%      \put(60,47){$\pi/3$}
      \end{overpic}
     \begin{figuretext}\label{DjsKdV.pdf}
       The domains $\{D_j\}_1^4$ in the complex $k$-plane with $D_1 = D_1' \cup D_1''$ and $D_4 = D_4' \cup D_4''$. 
     \end{figuretext}
     \end{center}
\end{figure}
For each $k \in \bar{D}_1 \cup \bar{D}_3$, the cubic polynomial 
$$4\nu^3 - \nu - (4k^3 - k)$$
vanishes at exactly one point in each of the three sets $\bar{D}_1'$, $\bar{D}_1''$, and $\bar{D}_3$. Denote these points by $\nu_1(k)$, $\nu_2(k)$, and $\nu_3(k)$, respectively. 

\begin{lemma}
The solution $u(x,t)$ for the initial-boundary-value problem  \eqref{auxiliary} for the linearized KdV equation (\ref{linearkdv}) has the representation 
\begin{align}\nonumber
 u(x,t)& = \frac{1}{\pi}\int_\R  e^{-2ikx-f(k) t} \hat{u}_0(k) dk
 + \frac{1}{2i\pi}\int_{\partial D_3}  e^{-2ikx-f(k)t} \bigg\{f'(k)\tilde{g}_0(f(k),t)
	\\ \label{linearsolution}
& - 2ik\frac{\hat{u}_0(\nu_2(k)) -\hat{u}_0(\nu_1(k))}{\nu_1(k) - \nu_2(k)}
 - 2i\frac{\nu_2(k) \hat{u}_0(\nu_1(k)) - \nu_1(k) \hat{u}_0(\nu_2(k))}{\nu_1(k) - \nu_2(k)}\bigg\}dk,
\end{align}
in terms of the initial data $u_0(x)$ and the Dirichlet data $g_0(t)$.  
Here, $f(k) = 2i(4k^3 - k)$ and
\begin{align*}
& \hat{u}_0(k) = \int_0^\infty e^{2ikx}u_0(x) dx, \qquad \im k \geq 0,
	\\
& \tilde{g}_0(\kappa,t) =  \int_0^t e^{\kappa s} g_0(s)ds, \hspace{1.2cm}  \kappa \in \C.
\end{align*}
\end{lemma}
\proofbegin
Equation (\ref{linearkdv}) is the compatibility condition of the Lax pair
\begin{align}\label{linearlax}
\begin{cases}
\varphi_x + 2ik \varphi = -\frac{i}{2k} u,
	\\
\varphi_t + f(k)\varphi = -2iku + u_x + \frac{i}{2k}(u + u_{xx}),
\end{cases}
\end{align}
where $k \in \C$ is the spectral parameter and $\varphi(x,t,k)$ is a scalar-valued eigenfunction. Write (\ref{linearlax}) in differential form as
$$d\bigl(e^{2ikx + f(k) t}\varphi\bigr)
= W,$$
where the closed one-form $W(x,t,k)$ is defined by
$$W = e^{2ikx + f(k)t}\biggl[-\frac{i}{2k}u \, dx + \Big(-2iku + u_x + \frac{i}{2k}(u + u_{xx})\Big)dt\biggr].$$
Stokes' Theorem implies that the integral of $W$ around the boundary of the domain $(0,\infty) \times (0,t)$ in the $(x,t)$-plane vanishes. This yields the global relation
\begin{align} \label{GR}
& \hat{u}_0(k) - e^{f(k)t}\hat{u}(k,t) + \tilde{g}(k,t) = 0, \qquad  k \in \bar{D}_1 \cup \bar{D}_2,
\end{align}
where 
\begin{align}\nonumber
& \tilde{g}(k,t) = (1- 4k^2) \tilde{g}_0(f(k),t) - 2ik\tilde{g}_1(f(k),t) + \tilde{g}_2(f(k),t),
	\\ \nonumber
& \hat{u}_0(k) = \int_0^\infty e^{2ikx}u_0(x) dx, \qquad
\hat{u}(k,t) = \int_0^\infty e^{2ikx} u(x,t) dx,
	\\\nonumber
& \tilde{g}_j(\kappa,t) =  \int_0^t e^{\kappa s} g_j(s)ds, \qquad j = 0,1,2.
\end{align}
Multiplying equation (\ref{GR}) by $\frac{1}{\pi}e^{-2ikx-f(k)t}$ and integrating the result along $\R$ with respect to $k$, it transpires that
\begin{align}\label{solutionformula}
 u(x,t) = &\; \frac{1}{\pi}\int_\R  e^{-2ikx-f(k)t} \hat{u}_0(k) dk
 - \frac{1}{\pi}\int_{\partial D_3}  e^{-2ikx-f(k) t} \tilde{g}(k,t)dk,
\end{align}
where Jordan's lemma has been used to deform the contour from $\R$ to $-\partial D_3$ in the second integral.

The final step consists of using the global relation to eliminate the two unknown functions $\tilde{g}_1(k,t)$ and $\tilde{g}_2(k,t)$ from (\ref{solutionformula}). Letting $k \to \nu_j(k)$, $j = 1,2$, in (\ref{GR}) gives
\begin{align}\nonumber
& \hat{u}_0(\nu_j(k)) - e^{f(k)t}\hat{u}(\nu_j(k),t)
+ (1- 4\nu_j^2(k)) \tilde{g}_0(f(k),t)  
	\\ \label{GRj}
&\hspace{1.54cm}- 2i\nu_j(k)\tilde{g}_1(f(k),t)+ \tilde{g}_2(f(k),t) = 0, \qquad k \in \bar{D}_3, \quad j = 1,2.
\end{align}
Solving these two equations for $\tilde{g}_1(f(k),t)$ and $\tilde{g}_2(f(k),t)$ leads to
\begin{align}\nonumber
\tilde{g}_1(f(k),t) = &\; \frac{i}{2\big(\nu_1(k) - \nu_2(k)\big)}\Bigl\{-\hat{u}_0(\nu_1(k))
+ \hat{u}_0(\nu_2(k)) 
	\\\nonumber
& + e^{f(k) t}[ \hat{u}(\nu_1(k), t) - \hat{u}(\nu_2(k),t)]
+ 4 \tilde{g}_0(f(k), t) (\nu_1^2(k) - \nu_2^2(k))\Bigr\},
	\\\nonumber
\tilde{g}_2(f(k),t) = &\; \frac{1}{\nu_1(k) - \nu_2(k)}\Bigl\{-\nu_1(k) \hat{u}_0(\nu_2(k))
+ \nu_2(k) \hat{u}_0(\nu_1(k)) 
	\\\nonumber
& + e^{f(k) t} [ \nu_1(k) \hat{u}(\nu_2(k), t)
- \nu_2(k)  \hat{u}(\nu_1(k),t) ]\Bigr\}
	\\\nonumber
& - \tilde{g}_0(f(k), t) (1 + 4\nu_1(k)\nu_2(k)).
\end{align}
Substituting these expressions into the solution formula (\ref{solutionformula}) and observing that, Jordan's lemma implies that the 
contributions from the terms involving $\{\hat{u}(\nu_j(k),t)\}_1^2$ vanish leads immediately to   (\ref{linearsolution}).
\proofend

Now suppose that $u_0(x) =0$ and that $g_0(t)$ is periodic with period $t_p$. 
To prove $(a)$, note that equation (\ref{linearsolution}) implies 
\begin{align*}
 u(x,t + t_p) - u(x,t) 
 = &\;  \frac{1}{2i\pi}\int_{\partial D_3} f'(k) e^{-2ikx-f(k)t} 
 \int_0^{t + t_p} e^{f(k) (s- t_p)} g_0(s)ds dk
 	\\
 & - \frac{1}{2i\pi}\int_{\partial D_3} f'(k)  e^{-2ikx-f(k)t} \int_0^t e^{f(k) s} g_0(s)ds dk.
\end{align*}
Making the change of variables $s \mapsto s + t_p$ in the part of the first $s$-integral that runs along $(t_p, t + t_p)$ and using the periodicity of $g_0$, 
there obtains
\begin{align*}
 u(x,t + t_p) - u(x,t) 
= &\; \frac{1}{2i\pi}\int_{\partial D_3} f'(k) e^{-2ikx-f(k)t} 
 \int_0^{t_p} e^{f(k)(s - t_p)} g_0(s)ds dk.
\end{align*}
Deforming the contour of integration from $\partial D_3$ to the steepest descent contour $\Gamma$, defined in Appendix \ref{steepestapp}, and see also  Figure 
\ref{Gamma.pdf}, a steepest descent argument yields (\ref{linearqdifference}) (see Proposition \ref{steepprop}).

To prove $(b)$,  define $m(x,t)$  by
\begin{align}\nonumber
m(x,t) := M(x,t + t_p) - M(x,t)
& = \int_t^{t + t_p} \partial_t M(x,t') dt'
	\\\nonumber
&   = - \int_t^{t + t_p} \int_x^\infty \big[u_x(x,t') + u_{xxx}(x,t') \big] dx dt'
  	\\ \label{Mdiff} 
& = \int_t^{t + t_p} \big[u(x,t') + u_{xx}(x,t')\big ] dt'.
\end{align}
Equation (\ref{linearsolution}) implies that
\begin{align*}\nonumber
u(x,t) = \frac{1}{2i\pi}  \int_{\partial \hat{D}_3} f'(k)  e^{-2ikx - f(k) t} \bigg(\frac{e^{f(k) t}}{f(k)} g_0(t)
- \int_0^{t} \frac{e^{f(k) s}}{f(k)} \dot{g}_0(s)ds\bigg) dk,
\end{align*}
where $\partial \hat{D}_3$ denotes the contour $\partial D_3$ deformed so that it passes to the right of the removable singularity at $k = 0$. 
\begin{figure}
\begin{center}
\begin{overpic}[width=.4\textwidth]{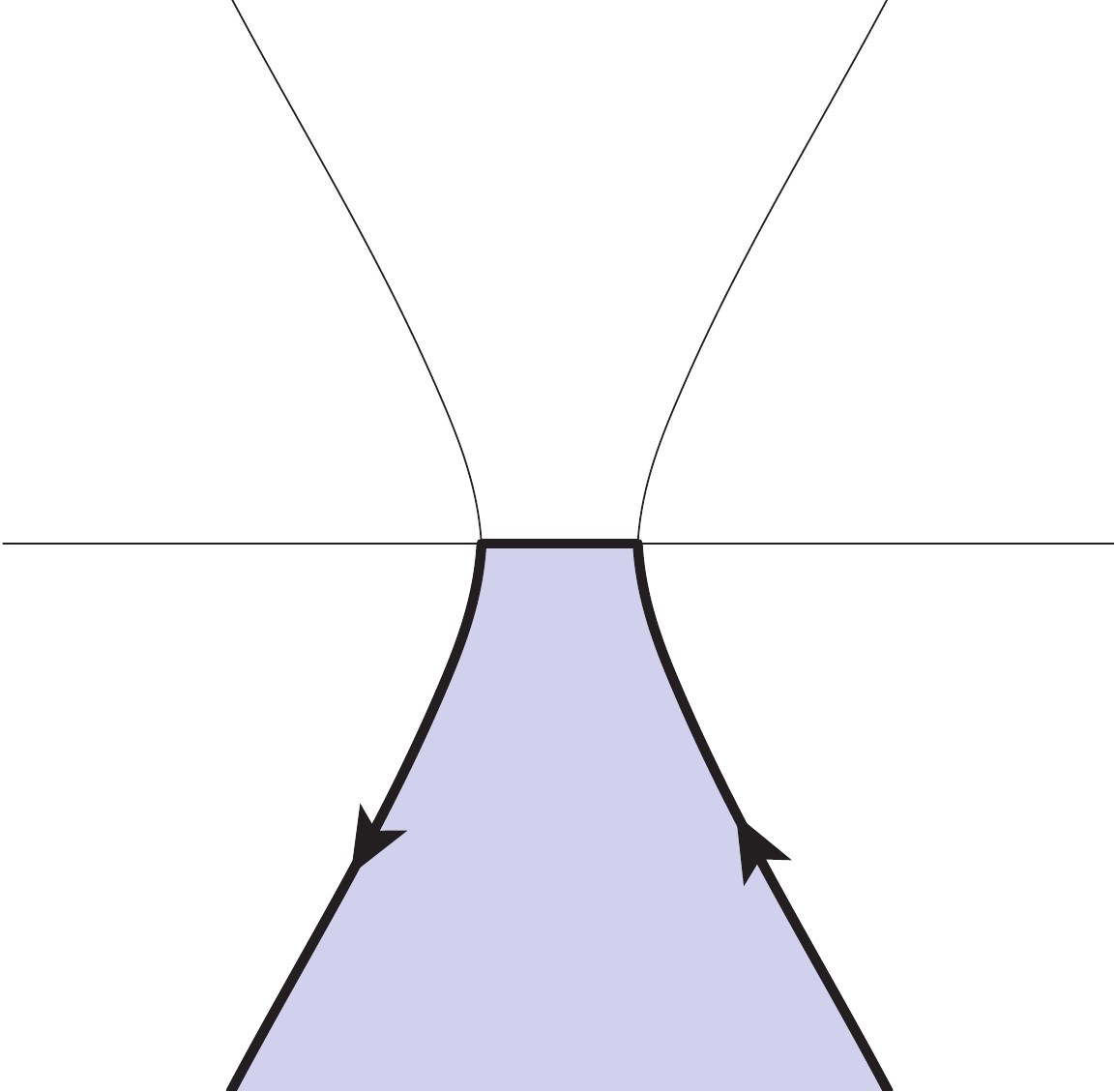}
       \put(71,23){$\partial D_3$}
      \put(47,12){$D_3$}
      \end{overpic}
      \qquad \quad
            \begin{overpic}[width=.4\textwidth]{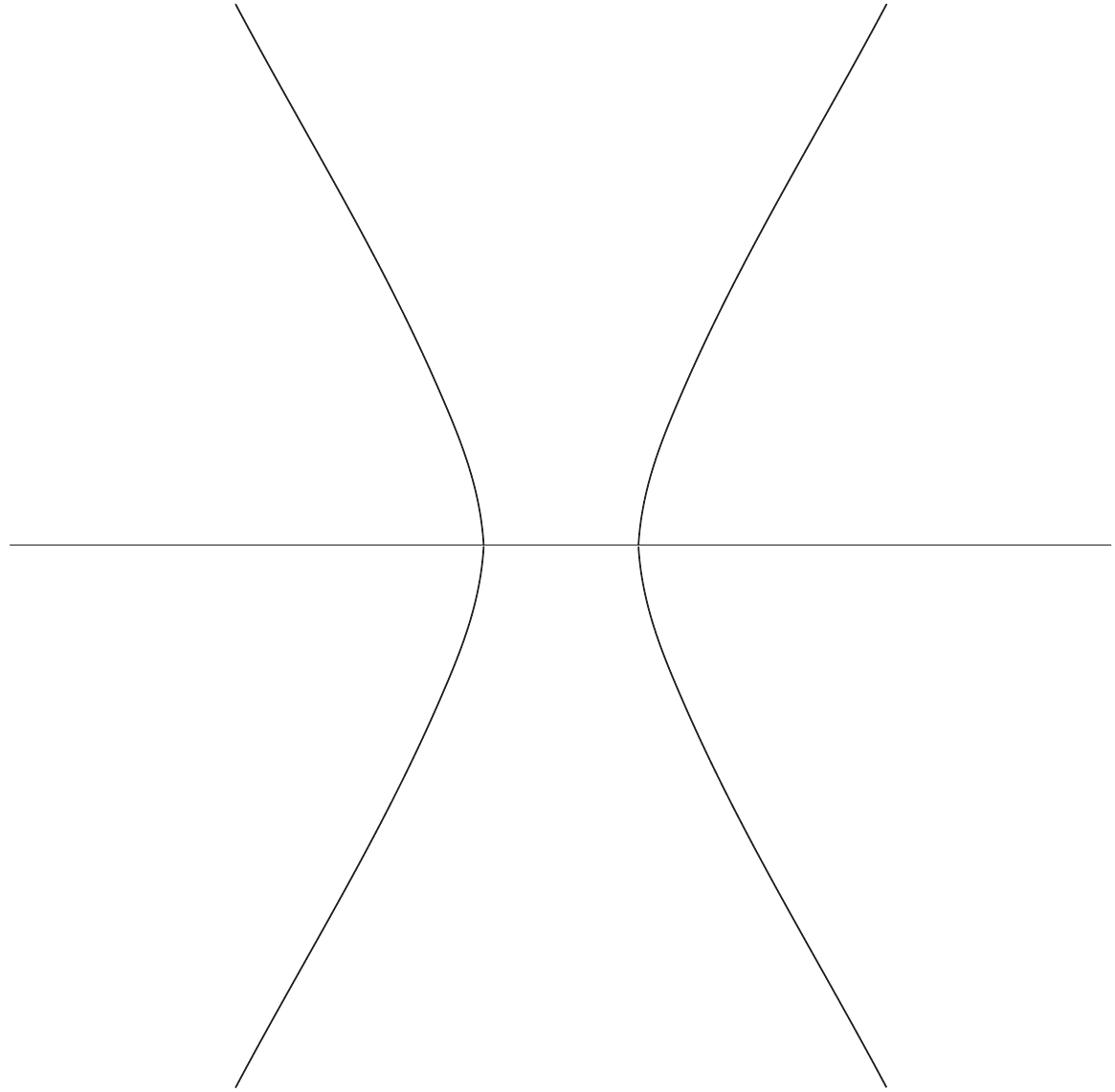}
      \put(72,64){$\re{f} < 0$}
      \put(38,80){$\re{f} > 0$}
      \put(5,64){$\re{f} < 0$}
      \put(5,30){$\re{f} > 0$}
      \put(38,15){$\re{f} < 0$}
      \put(72,30){$\re{f} > 0$}
      \end{overpic}
     \begin{figuretext}\label{D3hatfig}
       The contour $\partial D_3$ and the regions of definite sign of $\re f(k)$.
     \end{figuretext}
     \end{center}
\end{figure}
Since the contour has been deformed to $\partial \hat{D}_3$, the $k$-integral can be split and the part involving $g_0(t)$ can be calculated using Cauchy's theorem to reach the formula
\begin{align*}\nonumber
u(x,t) = &\;  \underset{k=0}{\res}  \frac{f'(k)  e^{-2ikx}}{f(k)} g_0(t)
- \frac{1}{2i\pi}  \int_{\partial \hat{D}_3} \frac{f'(k)}{f(k)}  e^{-2ikx - f(k) t}
 \int_0^{t} e^{f(k) s} \dot{g}_0(s)ds dk
 	\\
= &\; g_0(t) - \frac{1}{2i\pi}  \int_{\partial \hat{D}_3} \frac{f'(k)}{f(k)}  e^{-2ikx - f(k) t}
 \int_0^{t} e^{f(k) s} \dot{g}_0(s)ds dk.	
\end{align*}
Substituting this expression for $u(x,t)$ into (\ref{Mdiff}) gives
\begin{align}\nonumber
m(x,t) = & \int_t^{t + t_p} g_0(t') dt'
	\\ \label{mexpression1}
& - \frac{1}{4\pi}  \int_{\partial \hat{D}_3} \frac{f'(k)}{k}  e^{-2ikx} \bigg\{\int_t^{t + t_p}  e^{- f(k) t'}
 \int_0^{t'} e^{f(k) s} \dot{g}_0(s)ds dt' \bigg\}dk.
\end{align}
Integrating by parts with respect to $t'$ and using the periodicity of $g_0$, it is inferred that the curly bracket in (\ref{mexpression1}) equals
\begin{align*}\nonumber
& -\frac{e^{- f(k) (t + t_p)}}{f(k)} \int_0^{t + t_p} e^{f(k) s} \dot{g}_0(s)ds 
+ \frac{e^{- f(k) t}}{f(k)} \int_0^{t} e^{f(k) s} \dot{g}_0(s)ds 
+ \int_t^{t + t_p}  \frac{1}{f(k)} \dot{g}_0(t') dt'
 	\\
& \hspace{1cm}  = -\frac{e^{- f(k) (t + t_p)}}{f(k)} \int_0^{t_p} e^{f(k) s} \dot{g}_0(s)ds.
 \end{align*}
Hence, for any $x\geq 0$ and $t \geq 0$,
\begin{align}\nonumber
m(x,t) = & \int_t^{t + t_p} g_0(t') dt'
 + \frac{1}{4\pi}  \int_{\partial \hat{D}_3} \frac{f'(k)}{kf(k)}  e^{-2ikx} e^{- f(k) (t + t_p)} \int_0^{t_p} e^{f(k) s} \dot{g}_0(s)ds dk.
\end{align}
Deforming the contour of integration in the $k$-integral from $\partial \hat{D}_3$ to $\Gamma$, 
a steepest descent argument provides  (\ref{linearMdifference})
(see again Appendix A and Figure  \ref{Gamma.pdf}).

\section{Proof of Theorem \ref{mainth}}\nequation \label{nonlinearsec}
The proof relies on the  formulas\begin{subequations}\label{g11211222}
\begin{align}\label{g11start}
g_{11}(t) = &\; \frac{4}{\pi} \int_{\partial D_3} k\bigg[\hat{\chi}_{11}(t,k) + \frac{f'(k)}{f(k)}\frac{g_{01}(t)}{4}\bigg]dk,
	\\\label{g21start}
g_{21}(t) = &\;\frac{8}{\pi i}  \int_{\partial D_3} k^2 \bigg[\hat{\chi}_{11}(t,k) + \frac{f'(k)}{f(k)} \frac{g_{01}(t)}{4}\bigg]dk,
	\\\label{g12start}
g_{12}(t) = &\; \frac{4}{\pi} \int_{\partial D_3} k\bigg[\hat{\chi}_{12}(t,k) + \frac{f'(k)}{f(k)}\frac{g_{02}(t)}{4}\bigg]dk + \frac{2g_{01}(t)}{\pi} \int_{\partial \hat{D}_3} \chi_{21}(t,k) dk,
	\\ \nonumber
g_{22}(t) = &\;\frac{8}{\pi i}  \int_{\partial D_3} k^2\bigg[\hat{\chi}_{12}(t,k) + \frac{f'(k)}{f(k)} \frac{g_{02}(t)}{4}\bigg]dk 
	\\ \label{g22start}
& + \frac{2g_{11}(t)}{\pi} \int_{\partial \hat{D}_3} \chi_{21}(t,k)dk  - 2 g_{01}^2(t),
\end{align}
\end{subequations}
where $f(k) = 2i(4k^3 - k)$ as before and 
\begin{subequations}
\begin{align}\label{hatchi11}
 \hat{\chi}_{11}(t,k) = & -\frac{f'(k)}{4} \int_0^t e^{f(k)(t'-t)} g_{01}(t') dt',
	\\ \label{hatchi21def}
\chi_{21}(t,k) = &\; \frac{f'(k)}{2if(k)}\int_0^t (g_{01}(t') + g_{21}(t'))dt',
	\\\label{hatchi12def}
 \hat{\chi}_{12}(t,k) = & \sum_{m=1}^3 \nu_m(k) \nu_m'(k) \Phi_{12}(t, \nu_m(k)),
\end{align}
\end{subequations}
with
\begin{align*}	
\Phi_{12}(t,k) = &\int_0^t e^{f(k)(t'-t)} \bigg\{\frac{i}{2k}(g_{01} + g_{21})\Phi_{11} - 2ikg_{02} + g_{12} + \frac{i}{2k}(g_{02} + 2g_{01}^2 + g_{22})
	\\
& + \bigg(- 2ikg_{01} + g_{11} + \frac{i}{2k}(g_{01} + g_{21})\bigg)\Phi_{21}\bigg\} dt',
	\\
\Phi_{11}(t,k) = & \int_0^t e^{f(k)(t'-t)} \bigg(-2ikg_{01} + g_{11} + \frac{i}{2k}(g_{01} + g_{21})\bigg)dt',
	\\
\Phi_{21}(t,k) = & -\frac{i}{2k} \int_0^t (g_{01} + g_{21}) dt'.
\end{align*}	
These relations can be extracted from the nonlinear integral equations characterizing the Dirichlet to Neumann map of (\ref{kdv}) derived in \cite{LkdvD2N}.
The roots $\{\nu_j(k)\}_1^3$ satisfy the  identities
\begin{align*}
&\sum_{j=1}^3 \frac{1}{\nu_j(k)} = -\frac{2i}{f(k)}, \qquad
\sum_{j=1}^3 \frac{\nu_j'(k)}{\nu_j(k)} = \frac{f'(k)}{f(k)}, \qquad 
\sum_{j=1}^3 \nu_j^2(k)\nu_j'(k) = \frac{f'(k)}{8i},
	\\
&\sum_{j=1}^3 \nu_j(k) = \sum_{j=1}^3 \nu_j(k) \nu_j'(k)  = \sum_{j=1}^3 \nu_j^3(k)\nu_j'(k)= 0,
\qquad \sum_{j=1}^3 \nu_j^4(k)\nu_j'(k) = \frac{f'(k)}{32i}.
\end{align*}
In consequence, it transpires that 
\begin{align}\nonumber
\hat{\chi}_{12}(t,k) = &\; 
\frac{f'(k)}{4f(k)} \int_0^t e^{f(k)(t'-t)} (g_{01} + g_{21})(t') \int_0^{t'} (g_{01} + g_{21})(t'') dt'' dt'
	\\ \nonumber
& - \frac{f'(k)}{4f(k)} \int_0^t e^{f(k)(t'-t)} (g_{01} + g_{21})(t') \int_0^{t'} e^{f(k)(t''-t')}(g_{01} + g_{21})(t'') dt'' dt'
	\\ \nonumber
&- \frac{f'(k)}{4} \int_0^t e^{f(k)(t'-t)} g_{02}(t')  dt'
	\\ \nonumber
= &\; \frac{f'(k)}{4f(k)} \int_0^t e^{f(k)(t'-t)} (g_{01} + g_{21})(t') \bigg(\int_0^{t'} - \int_{t'}^t \bigg)(g_{01} + g_{21})(t'') dt'' dt'
	\\ \label{hatchi12start}
&- \frac{f'(k)}{4} \int_0^t e^{f(k)(t'-t)} g_{02}(t')  dt'.
\end{align}
Furthermore, since
$$\frac{f'(k)}{f(k)} = \frac{3}{k} + O(k^{-3}), \qquad {\rm as} \;\;k \to \infty,$$
Cauchy's theorem implies
\begin{align}\nonumber
\int_{\partial \hat{D}_3} \chi_{21}(t,k)dk
& = \bigg(-\frac{\pi}{2} + \pi \underset{k=0}{\res}\frac{f'(k)}{f(k)}\bigg)\int_0^t (g_{01}(t') + g_{21}(t'))dt'
	\\\label{intchi21}
& = \, \frac{\pi}{2}\int_0^t (g_{01}(t') + g_{21}(t'))dt'.
\end{align}

For each integer  $n \geq 1$, the third-order polynomial $f(k) + in\omega$ has three zeros; one zero in each of the three sets $\partial D_1'$, $\partial D_1''$ and $\partial D_3$.
Denote the unique solutions of $f(k) + in\omega = 0$ in $\partial D_3$ corresponding to $n = 1$ and $n = 2$ by $K$ and $L$, respectively. 

Several integrands will have singularities at points in the set $\{0, K, -\bar{K}, L, -\bar{L}\}$.  Let $\partial \hat{D}_3$ denote the contour $\partial D_3$ 
depicted in Figure \ref{D3hatKL.pdf} with indentations inserted so that $\partial \hat{D}_3$ passes to the right of the points $0, K, -\bar{K}, L$, and $-\bar{L}$, 
Assumption (\ref{omegaassumption}) implies that these indentations can be chosen to lie in $D_2 \cup D_4$.
Indeed, (\ref{omegaassumption}) implies that neither $K$ nor $L$ coincides with $\pm a \in \bar{D}_1 \cap \bar{D}_3$ where $a = \frac{1}{2\sqrt{3}}$. In fact, for $|\omega| < \frac{1}{3\sqrt{3}}$, both $K$ and $L$ belong to the interval $(-a, a)$.
For $\frac{1}{3\sqrt{3}} < |\omega| < \frac{2}{3\sqrt{3}}$, $K \in (-a,a)$ and $\im L < 0$. 
For $|\omega| > \frac{2}{3\sqrt{3}}$, $\im L < \im K < 0$. 
In particular, $K$ is a simple zero of $f(k) + i\omega = 0$ and $L$  is a simple zero of $f(k) + 2i\omega = 0$.
Similarly, $-\bar{K} \in \partial D_3$ is a simple zero of $f(k) - i\omega = 0$ and $-\bar{L} \in \partial D_3$   is a simple zero of $f(k) - 2i\omega = 0$.

\begin{figure}
\begin{center}
      \begin{overpic}[width=.5\textwidth]{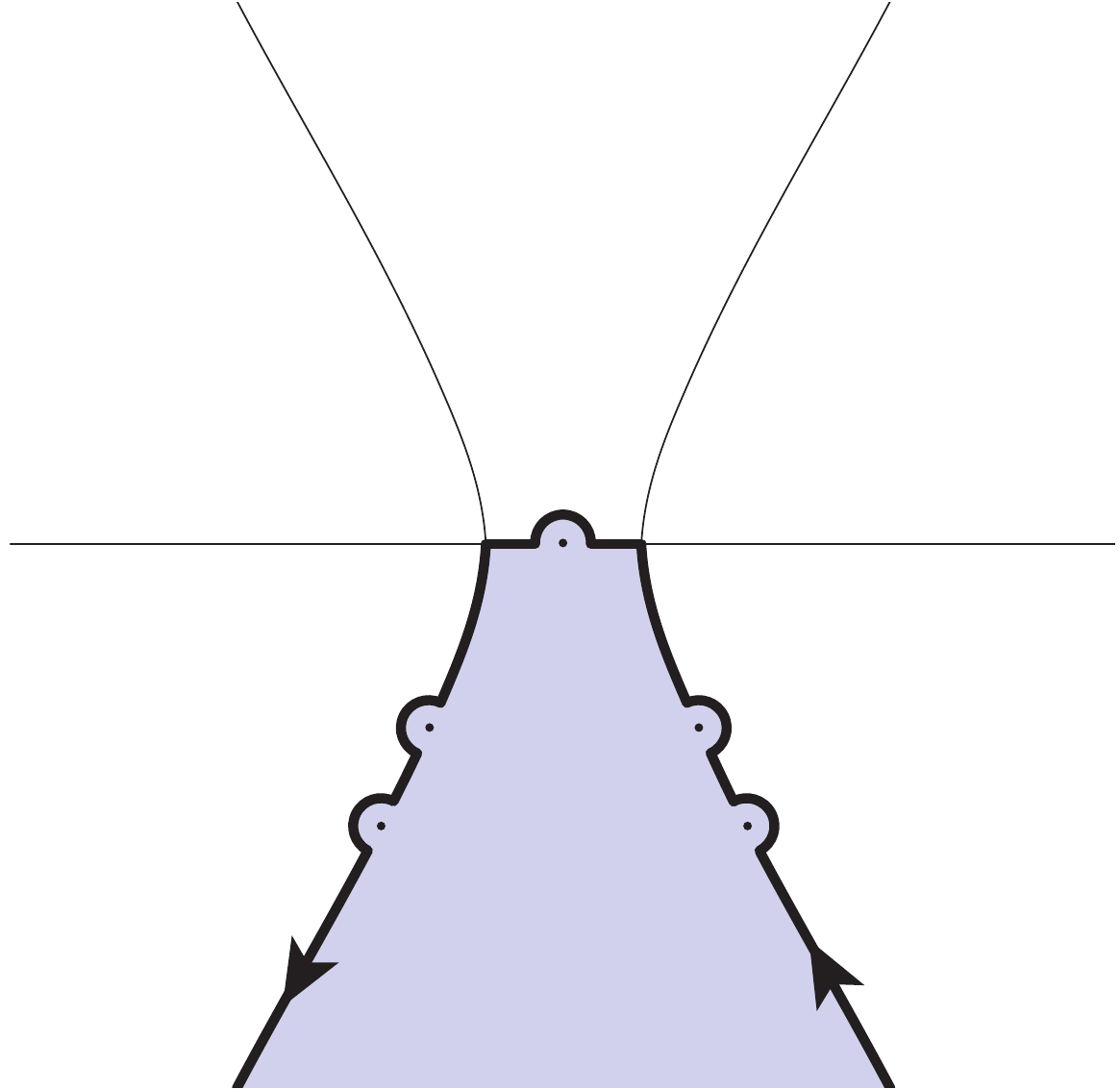}
      \put(77,11){$\partial \hat{D}_3$}
      \put(59,50){$a$}
      \put(34.5,50){$-a$}
      \put(55.5,29.5){$K$}
      \put(61,21){$L$}
      \put(39,29.5){$-\bar{K}$}
      \put(35,21){$-\bar{L}$}
      \put(49,43){$0$}
      \end{overpic}
     \begin{figuretext}\label{D3hatKL.pdf}
       The contour $\partial \hat{D}_3$ in the case of $\omega > \frac{2}{3\sqrt{3}}$.
     \end{figuretext}
     \end{center}
\end{figure}

The discussion continues with the particular choice $g_0(t) = \epsilon \sin(\omega t)$.   Each of the expressions in (\ref{g11211222}) is 
considered in turn with this choice of boundary data.  
\subsection{Asymptotics of $g_{11}(t)$}
Suppose $g_0(t) = \epsilon \sin(\omega t)$. Equation (\ref{hatchi11}) becomes
\begin{align}\nonumber
 \hat{\chi}_{11}(t,k) = & -\frac{f'(k)}{8i} \int_0^t e^{f(k)(t'-t)} (e^{i\omega t'} - e^{-i\omega t'})dt'
 	\\ \label{hatchi11computed}
 = & -\frac{f'(k)}{8i}\bigg[\frac{e^{i\omega t}}{f(k) + i\omega} - \frac{e^{-i\omega t}}{f(k) - i\omega}\bigg]
 + A_1(t,k),
\end{align}
where
\begin{align}\label{A1def}
A_1(t,k) = \frac{f'(k)}{8i} e^{-f(k)t} \bigg[\frac{1}{f(k) + i\omega} - \frac{1}{f(k) - i\omega}\bigg].
\end{align}
Substituting this into (\ref{g11start}) and using Cauchy's theorem gives the expression 
\begin{align}\nonumber
g_{11}(t) = &\; \frac{4}{\pi}  \int_{\partial \hat{D}_3} \frac{k}{8i} \bigg(-\frac{f'(k)}{f(k) + i\omega} + \frac{f'(k)}{f(k)}\bigg)dk e^{i\omega t} 
	\\ \nonumber
& + \frac{4}{\pi} \int_{\partial \hat{D}_3}  \frac{k}{8i} \bigg(\frac{f'(k)}{f(k) - i\omega} -  \frac{f'(k)}{f(k)}\bigg) dk e^{-i\omega t}
 + \frac{4}{\pi} \int_{\partial \hat{D}_3}  kA_1(t,k)dk
 	\\ \label{g11computed}
= &\;  - K  e^{i\omega t} 
 -\bar{K}  e^{-i\omega t}
 + \frac{4}{\pi}  \int_{\partial \hat{D}_3}  kA_1(t,k)dk,
\end{align}
where the formulas
$$\underset{k = K}{\res} \frac{kf'(k)}{f(k) + i\omega} = K \;\; \quad {\rm and} \quad \;\;
 \underset{k = -\bar{K}}{\res} \frac{kf'(k)}{f(k) - i\omega} = -\bar{K}
 $$
have been applied.  Using Jordan's lemma to deform the contour $\partial \hat{D}_3$ to the steepest descent contour $\Gamma$ 
depicted in  Figure \ref{Gamma.pdf}, Proposition \ref{steepprop} reveals that
\begin{align}\label{intkA1}
\int_{\partial \hat{D}_3}  kA_1(t,k)dk = O\big(t^{-\frac{3}{2}}\big), \qquad t\to \infty, 
\end{align}
thereby establishing  (\ref{g11end}).

\subsection{Asymptotics of $g_{21}(t)$}
Substituting (\ref{hatchi11computed}) into (\ref{g21start}) and using Cauchy's theorem provides the formula
\begin{align}\nonumber
g_{21}(t) = &\;\frac{8}{\pi i}  \int_{\partial \hat{D}_3}
\frac{k^2}{8i} \bigg[-\frac{f'(k)}{f(k) + i\omega} + \frac{f'(k)}{f(k)} \bigg]dke^{i\omega t}
	\\\nonumber
& + \frac{8}{\pi i}  \int_{\partial \hat{D}_3}\frac{k^2}{8i}\bigg[\frac{f'(k)}{f(k) - i\omega} - \frac{f'(k)}{f(k)}\bigg]dk e^{-i\omega t} 
+ \frac{8}{\pi i}  \int_{\partial \hat{D}_3}k^2A_1(t,k)dk
	\\ \label{g21computed}
= &\;2iK^2 e^{i\omega t}
- 2i \bar{K}^2 e^{-i\omega t} 
+ \frac{8}{\pi i}  \int_{\partial \hat{D}_3}k^2A_1(t,k)dk.	
\end{align}
when $g_0(t) = \epsilon \sin(\omega t)$.  Here, we have used that 
$$\underset{k = K}{\res} \frac{k^2f'(k)}{f(k) + i\omega} = K^2, \qquad
\underset{k = -\bar{K}}{\res} \frac{k^2f'(k)}{f(k) - i\omega} = \bar{K}^2.$$
Just as in the case of $g_{11}(t)$, a steepest descent argument shows  that
$$\int_{\partial \hat{D}_3}  k^2 A_1(t,k)dk = O\big(t^{-\frac{3}{2}}\big), \qquad t\to \infty, $$
which  proves (\ref{g21end}).

\subsection{Asymptotics of $g_{12}(t)$}
The computation of $g_{12}(t)$ relies on (\ref{g12start}) and proceeds via a series of lemmas.

\begin{lemma} The integral in the last term on the right-hand side of (\ref{g12start}) is given by
\begin{align} \label{chi21int}
& \int_{\partial \hat{D}_3} \chi_{21}(t,k) dk  = - \frac{\pi e^{i\omega t}}{8K}
- \frac{\pi e^{-i\omega t} }{8\bar{K}} 
+ \frac{\pi}{4} \re\Big(\frac{1}{K}\Big) - 4i \int_0^t  \int_{\partial \hat{D}_3}k^2A_1(t',k)dk dt'.
\end{align}
\end{lemma}
\proofbegin
Equation (\ref{g21computed}) together with the identities
$$1 - 4K^2  = \frac{\omega}{2K}, \;\;\quad {\rm and} \;\;\quad1 - 4\bar{K}^2 = \frac{\omega}{2\bar{K}},$$
imply that
\begin{align}\label{g01plusg21}
g_{01}(t) + g_{21}(t) = \frac{\omega}{4iK}e^{i\omega t}
- \frac{\omega}{4i\bar{K}} e^{-i\omega t}
+ \frac{8}{\pi i}  \int_{\partial \hat{D}_3}k^2A_1(t,k)dk.
\end{align}
After substituting this expression into (\ref{intchi21}), formula (\ref{chi21int}) emerges.
\proofend

\begin{lemma} The function $\hat{\chi}_{12}(t,k)$  in (\ref{hatchi12def}) has the representation
\begin{align}\nonumber
\hat{\chi}_{12}(t,k) = &\;\frac{\omega f'(k)}{64if(k)} \bigg\{\frac{1}{K^2}\bigg(- \frac{2e^{2i\omega t}}{f(k) + 2i\omega}
 + \frac{e^{i\omega t} }{f(k) + i\omega}
 + \frac{e^{2i\omega t}}{f(k) + i\omega}\bigg)
	\\\label{hatchi12}
&  + \frac{1}{|K|^2}\bigg(\frac{e^{i\omega t}}{f(k) + i\omega} 
+ \frac{1}{f(k) + i\omega}
- \frac{e^{-i\omega t}}{f(k) - i\omega} 
- \frac{1}{f(k) - i\omega}\bigg)
	\\ \nonumber
& + \frac{1}{\bar{K}^2}\bigg(\frac{2e^{-2i\omega t}}{f(k) - 2i\omega} 
- \frac{e^{-i\omega t}}{f(k) - i\omega} 
- \frac{e^{-2i\omega t}}{f(k) - i\omega}\bigg)\bigg\}
+ E(t,k) + A_2(t,k),	
\end{align}
where $E(t,k)$ and $A_2(t,k)$ are given by
\begin{align}\nonumber
E(t,k) =  & \; \frac{f'(k)}{4f(k)} \int_0^t e^{f(k)(t'-t)} \bigg\{ \frac{8}{\pi i}  \int_{\partial \hat{D}_3}k^2A_1(t',k)dk\bigg\}\bigg\{-\frac{2e^{i\omega t'} - 1 - e^{i\omega t}}{4K}
	\\ \nonumber
& - \frac{2e^{-i\omega t'} - 1 - e^{-i\omega t}}{4\bar{K}}
 + \frac{8}{\pi i} \bigg(\int_0^{t'} - \int_{t'}^t \bigg) \int_{\partial \hat{D}_3}k^2A_1(t'',k)dkdt''\bigg\} dt'
	\\ \nonumber
&+ \frac{f'(k)}{4f(k)} \int_0^t e^{f(k)(t'-t)} \bigg\{\frac{\omega}{4iK} e^{i\omega t'}
- \frac{\omega}{4i\bar{K}} e^{-i\omega t'}\bigg\}
	\\ \label{Edef}
& \times \bigg\{\frac{8}{\pi i}  \bigg(\int_0^{t'} - \int_{t'}^t\bigg) \int_{\partial \hat{D}_3}k^2A_1(t'',k)dkdt''\bigg\} dt'
\end{align}
and
\begin{align}\nonumber
A_2(t,k) = &\; \frac{\omega f'(k)}{64if(k)} \bigg\{\frac{1}{K^2}\bigg(
 \frac{2e^{-f(k)t}}{f(k) + 2i\omega}
 - \frac{e^{-f(k)t}}{f(k) + i\omega}
 - \frac{e^{i\omega t - f(k)t}}{f(k) + i\omega}\bigg)
	\\\nonumber
&  + \frac{1}{|K|^2}\bigg(-\frac{e^{-f(k)t}}{f(k) + i\omega} 
- \frac{e^{-i\omega t - f(k)t}}{f(k) + i\omega}
+ \frac{e^{-f(k)t}}{f(k) - i\omega} 
+ \frac{e^{i\omega t -f(k)t}}{f(k) - i\omega}\bigg)
	\\\label{A2def}
& + \frac{1}{\bar{K}^2}\bigg(-\frac{2e^{-f(k)t}}{f(k) - 2i\omega} 
+ \frac{ e^{-f(k)t}}{f(k) - i\omega} 
+ \frac{e^{-i\omega t -f(k)t}}{f(k) - i\omega}\bigg)\bigg\}.
\end{align}
\end{lemma}
\proofbegin
In view of (\ref{g01plusg21}), it transpires that 
\begin{align*}
\bigg(\int_0^{t'} - \int_{t'}^t\bigg)(g_{01}(t'') + g_{21}(t'')) dt'' =& -\frac{2e^{i\omega t'} - 1 - e^{i\omega t}}{4K}
- \frac{2e^{-i\omega t'} - 1 - e^{-i\omega t}}{4\bar{K}}
	\\
& + \frac{8}{\pi i} \bigg(\int_0^{t'} - \int_{t'}^t\bigg) \int_{\partial \hat{D}_3}k^2A_1(t'',k)dkdt''.
\end{align*}
Inserting this into the expression (\ref{hatchi12start}) for $\hat{\chi}_{12}(t,k)$ leads to 
\begin{align*}
\hat{\chi}_{12}(t,k) = &\; \frac{\omega f'(k)}{64if(k)} \int_0^t e^{f(k)(t'-t)} 
 \bigg\{- \frac{2e^{2i\omega t'} - e^{i\omega t'} - e^{i\omega (t+t')}}{K^2}
	\\
& 
 + \frac{ e^{i\omega t'} + e^{-i\omega (t-t')} - e^{-i\omega t'} - e^{i\omega (t-t')}}{|K|^2}
+ \frac{2e^{-2i\omega t'} - e^{-i\omega t'} - e^{-i\omega (t+t')}}{\bar{K}^2}\bigg\} dt'
	\\
& + E(t,k),
\end{align*}
where  $E(t,k)$ is as in (\ref{Edef}). Computing the integrals with respect to $t'$ gives (\ref{hatchi12}). 
\proofend

\begin{lemma}\label{khatchi12lemma} The first term on the right-hand side of (\ref{g12start}) is 
\begin{align}\nonumber
 \frac{4}{\pi}  \int_{\partial D_3} k\hat{\chi}_{12}(t,k) dk
= & \; \frac{1}{8i}\bigg(\frac{L}{K^2} - \frac{1}{K}\bigg) e^{2i\omega t} 
+ \frac{i}{4}\re\bigg(\frac{1}{K}\bigg) e^{i\omega t}+ \frac{1}{4}  \im \frac{1}{K}
	\\ \nonumber
& 
-\frac{i}{4}\re\bigg(\frac{1}{K}\bigg) e^{-i\omega t}
 - \frac{1}{8i}\bigg(\frac{\bar{L}}{\bar{K}^2} - \frac{1}{\bar{K}}\bigg) e^{-2i\omega t} 
	\\ \label{khatchi12int}
& + \frac{4}{\pi}  \int_{\partial \hat{D}_3} k[E(t,k) + A_2(t,k)]dk.
\end{align}
\end{lemma}
\proofbegin
This is a consequence of the expression (\ref{hatchi12}) for $\hat{\chi}_{12}(t,k)$ and Cauchy's theorem. For example, (\ref{hatchi12}) implies that the coefficient of $e^{2i\omega t}$ is 
$$\frac{\omega}{16 i \pi K^2} \int_{\partial \hat{D}_3} \frac{k f'(k)}{f(k)} \bigg(\frac{1}{f(k) + i\omega} - \frac{2}{f(k) + 2i\omega}\bigg)dk.$$
Since the integrand has simple poles at $k = K$ and $k = L$, Cauchy's theorem implies that the latter integral equals
\begin{align*}
 \frac{\omega}{8 K^2} \Big(\underset{k=K}{\res} + \underset{k=L}{\res}\Big) \frac{k f'(k)}{f(k)} \bigg(\frac{1}{f(k) + i\omega} - \frac{2}{f(k) + 2i\omega}\bigg) = \frac{1}{8i}\bigg(\frac{L}{K^2} - \frac{1}{K}\bigg).
\end{align*}	
\proofend

Using (\ref{chi21int}) and (\ref{khatchi12int}) in the expression (\ref{g12start}) for $g_{12}(t)$, one finds that
\begin{align}\label{g12computed}
g_{12}(t) = &\; \frac{1}{8i}\bigg(\frac{L}{K^2} - \frac{2}{K}\bigg)e^{2i\omega t}
+\frac{1}{2}  \im\bigg( \frac{1}{K}\bigg)
- \frac{1}{8i} \bigg(\frac{ \bar{L}}{\bar{K}^2}- \frac{2}{\bar{K}}\bigg)e^{-2i\omega t} + F_1(t,k),
\end{align}
where
\begin{align}\label{F1def}
F_1(t,k) =  - 4\frac{e^{i\omega t} - e^{-i\omega t}}{\pi}\int_0^t  \int_{\partial \hat{D}_3}k^2A_1(t',k)dk dt'   \nonumber\\
+ \frac{4}{\pi}  \int_{\partial \hat{D}_3} k[E(t,k) + A_2(t,k)]dk.
\end{align}
To complete the proof of (\ref{g12end}), it is enough to show that $F_1(t,k)$ is $O(t^{-3/2})$ 
as $t \to \infty$.   The proof of this fact is postponed to Section \ref{F1F2subsec}.

\subsection{Asymptotics of $g_{22}(t)$}
The computation of $g_{22}(t)$ relies on (\ref{g22start}).

\begin{lemma} The last two terms on the right-hand side of (\ref{g22start}) can be written as
\begin{align}\nonumber
 &\frac{2g_{11}(t)}{\pi} \int_{\partial \hat{D}_3} \chi_{21}(t,k)dk  - 2 g_{01}^2(t) 
	\\\nonumber
  &  \hspace{.6cm}  = \; \frac{3}{4} e^{2i\omega t}
- \frac{1}{4}\bigg(1 + \frac{K}{\bar{K}} \bigg) e^{i\omega t}
+ \frac{1}{2}\re\bigg(\frac{K}{\bar{K}}\bigg)  -1
- \frac{1}{4}\bigg(1 + \frac{\bar{K}}{K} \bigg) e^{-i\omega t}
+ \frac{3}{4} e^{-2i\omega t}
	\\\nonumber
& \hspace{1.2cm} + \frac{8i}{\pi}\big( K  e^{i\omega t} + \bar{K}  e^{-i\omega t}\big) 
  \int_0^t  \int_{\partial \hat{D}_3}k^2A_1(t',k)dk dt'
	\\\nonumber
&  \hspace{1.2cm} + \frac{1}{\pi}  \bigg\{2 \re\bigg(\frac{1}{K}\bigg) - \frac{1}{K}e^{i\omega t}
- \frac{1}{\bar{K}} e^{-i\omega t} \bigg\} \int_{\partial \hat{D}_3}  kA_1(t,k)dk
	\\ \label{g22contribution1}
&  \hspace{1.2cm} - \frac{32i}{\pi^2} \bigg( \int_{\partial \hat{D}_3}  kA_1(t,k)dk\bigg)\bigg( \int_0^t  \int_{\partial \hat{D}_3}k^2A_1(t',k)dk dt'\bigg).
\end{align}
\end{lemma}
\proofbegin
This follows from the expression (\ref{g11computed}) for $g_{11}(t)$, the expression (\ref{chi21int}) for $\int_{\partial \hat{D}_3} \chi_{21}(t,k)dk$, and the fact that $g_{01}(t) = \sin{\omega t}$.  
\proofend

\begin{lemma}The first term on the right-hand side of (\ref{g22start}) is given by
\begin{align}\nonumber
 \frac{8}{\pi i}  \int_{\partial D_3} k^2\hat{\chi}_{12}(t,k) dk
= & \; \frac{1}{4}\bigg(1 - \frac{L^2}{K^2}\bigg) e^{2i\omega t} 
+ \frac{1}{4}\bigg(1 + \frac{K}{\bar{K}}\bigg) e^{i\omega t}
	\\ \nonumber
& + \frac{1}{2}  \re\bigg(\frac{K}{\bar{K}}\bigg)
 + \frac{1}{4}\bigg(1 + \frac{\bar{K}}{K}\bigg) e^{-i\omega t}
 + \frac{1}{4}\bigg(1 - \frac{\bar{L}^2}{\bar{K}^2}\bigg) e^{-2i\omega t} 
	\\ \label{g22contribution2}
& + \frac{8}{\pi i}  \int_{\partial \hat{D}_3} k^2[E(t,k) + A_2(t,k)]dk.
\end{align}
\end{lemma}
\proofbegin
As in the proof of Lemma \ref{khatchi12lemma}, this is a consequence of the expression (\ref{hatchi12}) for $\hat{\chi}_{12}(t,k)$ and Cauchy's theorem. 
\proofend

According to (\ref{g22start}), $g_{22}(t)$ is the sum of the expressions in (\ref{g22contribution1}) and (\ref{g22contribution2}). Thus,
\begin{align}\label{g22computed}
g_{22}(t) =  \bigg(1 - \frac{L^2}{4K^2}\bigg) e^{2i\omega t} 
+ \re\bigg(\frac{K}{\bar{K}}\bigg) -1  
+ \bigg(1 - \frac{\bar{L}^2}{4\bar{K}^2}\bigg) e^{-2i\omega t} + F_2(t,k),
\end{align}
where
\begin{align}\nonumber
F_2(t,k) = &\;  \frac{8i}{\pi}\big( K  e^{i\omega t} + \bar{K}  e^{-i\omega t}\big) 
  \int_0^t  \int_{\partial \hat{D}_3}k^2A_1(t',k)dk dt'
	\\\nonumber
& + \frac{1}{\pi}  \bigg\{2 \re\bigg(\frac{1}{K}\bigg) - \frac{1}{K}e^{i\omega t}
- \frac{1}{\bar{K}} e^{-i\omega t} \bigg\} \int_{\partial \hat{D}_3}  kA_1(t,k)dk
	\\\nonumber
&- \frac{32i}{\pi^2} \bigg( \int_{\partial \hat{D}_3}  kA_1(t,k)dk\bigg)\bigg( \int_0^t  \int_{\partial \hat{D}_3}k^2A_1(t',k)dk dt'\bigg)
	\\\label{F2def}
& + \frac{8}{\pi i}  \int_{\partial \hat{D}_3} k^2[E(t,k) + A_2(t,k)]dk.
\end{align}
The proof of (\ref{g22end}) is completed by establishing that $F_2(t,k)$ is $O(t^{-\frac{3}{2}})$ in the next subsection. 

\subsection{Asymptotics of $F_1(t,k)$ and $F_2(t,k)$}\label{F1F2subsec}
In this subsection,  the proof of Theorem \ref{mainth} is completed by showing that $F_j(t,k)$ is of order $O(t^{-3/2})$ as $t \to \infty$, 
$j= 1, 2$.

\begin{lemma} The functions $A_1(t,k)$ and $A_2(t,k)$ defined in (\ref{A1def}) and (\ref{A2def}) satisfy
\begin{align}  \label{intk2A1}
\int_0^t \int_{\partial \hat{D}_3} k^2 A_1(t',k) dk dt'
= \frac{\pi i}{2\omega} \re(K^2) + O\big(t^{-\frac{3}{2}}\big), \qquad t\to \infty,
\end{align}
and
\begin{align}\label{intkjA2}
\int_{\partial \hat{D}_3} k^j A_2(t,k) dk = O\big(t^{-\frac{3}{2}}\big), \hspace{1.7cm}  t\to \infty, \quad j = 1,2.
\end{align}
\end{lemma}
\begin{proof}
A computation using (\ref{A1def}), Proposition \ref{steepprop}, and Cauchy's theorem shows that the left-hand side of (\ref{intk2A1}) equals
\begin{align}\nonumber
  \int_{\partial \hat{D}_3} k^2 & \frac{f'(k)}{8i} \bigg[\frac{1}{f(k) + i\omega} - \frac{1}{f(k) - i\omega}\bigg] \int_0^t e^{-f(k)t'} dt' dk 
 	\\\nonumber
 = &\; -\frac{1}{8i} \int_{\partial \hat{D}_3} \frac{k^2f'(k)}{f(k)} \bigg[\frac{1}{f(k) + i\omega} - \frac{1}{f(k) - i\omega}\bigg] e^{-f(k)t}dk 
 	\\\nonumber
& \; \, + \frac{1}{8i} \int_{\partial \hat{D}_3} \frac{k^2f'(k)}{f(k)}  \bigg[\frac{1}{f(k) + i\omega} - \frac{1}{f(k) - i\omega}\bigg] dk
	\\\nonumber
 = &\; \; O\big(t^{-\frac{3}{2}}\big) + \frac{\pi}{4} \bigg\{\underset{k = K}{\res} \frac{k^2f'(k)}{f(k)}\frac{1}{f(k) + i\omega} - \underset{k = -\bar{K}}{\res}\frac{k^2f'(k)}{f(k)}\frac{1}{f(k) - i\omega}\bigg\}	
	\\\label{int0tk2A1int}
 = & \; \;\;\frac{\pi i}{4\omega} (K^2 + \bar{K}^2) + O\big(t^{-\frac{3}{2}}\big). 
\end{align}
This proves (\ref{intk2A1}). Recalling the definition (\ref{A2def}) of $A_2$ and deforming the contour to $\Gamma$, equation (\ref{intkjA2}) follows immediately from Proposition \ref{steepprop}.
\end{proof}

\begin{lemma} \label{Elemma}
The function $E(t,k)$ defined in (\ref{Edef}) satisfies
\begin{subequations}\label{intkE}
\begin{align}
&\frac{4}{\pi} \int_{\partial \hat{D}_3} k E(t,k) dk = \frac{2i}{\omega}\re(K^2) (e^{i\omega t} - e^{-i\omega t}) +  O\big(t^{-\frac{3}{2}}\big), \hspace{1.7cm}  t  \to \infty,
	\\
& \frac{8}{\pi i} \int_{\partial \hat{D}_3} k^2 E(t,k) dk
= \frac{4}{\omega} \re(K^2) (Ke^{i\omega t} + \bar{K} e^{-i\omega t}) + O\big(t^{-\frac{3}{2}}\big),  \qquad t  \to \infty.
\end{align}
\end{subequations}
\end{lemma}
\proofbegin
First note that the definition (\ref{A1def}) of $A_1$ and Cauchy's theorem yield
\begin{align*}
&\bigg(\int_0^{t'} - \int_{t'}^t\bigg) \int_{\partial \hat{D}_3}k^2A_1(t'',k)dkdt''
	\\
&= \frac{\pi i}{4\omega} (K^2 + \bar{K}^2) - \int_{\partial \hat{D}_3}k^2\frac{f'(k)}{8i} 
 \frac{2e^{-f(k)t'} - e^{-f(k)t}}{f(k)}\bigg[\frac{1}{f(k) + i\omega} - \frac{1}{f(k) - i\omega}\bigg]  dk.
\end{align*}
Substituting this and the expression (\ref{A1def}) for $A_1$ into the definition (\ref{Edef}) of $E(t,k)$ and performing the integrals with respect to $t'$ provides the formula
\begin{align}\nonumber
&E(t,k) = -\frac{f'(k)}{4\pi f(k)}  \int_{\partial \hat{D}_3} k'^2 f'(k') \bigg[\frac{1}{f(k') + i\omega} - \frac{1}{f(k') - i\omega}\bigg] 
	\\\nonumber
& \times \bigg\{-\frac{e^{i\omega t-f(k')t} - e^{-f(k)t}}{2K(f(k) - f(k') + i\omega)} 
+ \frac{e^{-f(k')t} - e^{-f(k)t}}{4K(f(k) - f(k'))} 
+ \frac{e^{i\omega t -f(k')t} - e^{i\omega t -f(k)t}}{4K(f(k) - f(k'))}
	\\\nonumber
& - \frac{e^{-i\omega t-f(k')t} - e^{-f(k)t}}{2\bar{K}(f(k)- f(k')-i\omega)} 
+ \frac{e^{-f(k')t} - e^{-f(k)t}}{4\bar{K}(f(k)-f(k'))} 
+ \frac{ e^{-i\omega t -f(k')t} - e^{-i\omega t - f(k)t}}{4\bar{K}(f(k) - f(k'))}
	\\\nonumber
& +4 \frac{e^{-f(k')t} - e^{-f(k)t}}{\omega(f(k) - f(k'))} \re(K^2)\bigg\} dk'
- \frac{f'(k)}{4f(k)}  \frac{1}{\pi^2}  \int_{\partial \hat{D}_3} k'^2 f'(k')
	\\\nonumber
&\times  \bigg[\frac{1}{f(k') + i\omega} - \frac{1}{f(k') - i\omega}\bigg] \int_{\partial \hat{D}_3} k''^2f'(k'')
\bigg[\frac{1}{f(k'') + i\omega} - \frac{1}{f(k'') - i\omega}\bigg]  
	\\\nonumber
&\times
\bigg\{ 2\frac{e^{-f(k'')t -f(k')t} - e^{-f(k)t}}{f(k'')(f(k) -f(k') - f(k''))} 
 - \frac{e^{-f(k'')t -f(k')t} - e^{-f(k'')t-f(k)t}}{f(k'')(f(k) - f(k'))} \bigg\} dk''dk'
	\\\nonumber
& + \frac{\omega f'(k)}{4\pi f(k)} \int_{\partial \hat{D}_3}  k'^2\frac{f'(k')}{f(k')}  \bigg[\frac{1}{f(k') + i\omega} - \frac{1}{f(k') - i\omega}\bigg]
	\\\nonumber
&\times
\bigg\{
\frac{e^{-f(k')t + i\omega t} - e^{-f(k)t}}{2iK(f(k) - f(k') + i\omega)}  
-\frac{e^{-f(k')t + i\omega t} - e^{-f(k')t-f(k)t}}{4iK(f(k) + i\omega)}  
	\\ \nonumber
& - \frac{e^{-f(k')t-i\omega t} - e^{-f(k)t} }{2i\bar{K}(f(k) - f(k') -i\omega)}
+ \frac{e^{-f(k')t-i\omega t} - e^{-f(k')t - f(k)t}}{4i\bar{K}(f(k) -i\omega)}
\bigg\}dk'
	\\\label{Eexpression}
&+ \re(K^2)  \frac{f'(k)}{f(k)}
 \bigg\{ \frac{e^{i\omega t} - e^{-f(k)t}}{4iK(f(k) + i\omega)} 
- \frac{e^{-i\omega t} - e^{-f(k)t}}{4i\bar{K}(f(k) -i\omega)} \bigg\}.
\end{align}
To prove (\ref{intkE}),  multiply (\ref{Eexpression}) by $k^j$, $j = 1,2$, and integrate over $\partial \hat{D}_3$. We claim that all the terms on 
the right-hand side of \eqref{Eexpression} which involve a factor of $e^{-f(k) t}$, $e^{-f(k') t}$, or $e^{-f(k'') t}$ are of order $O(t^{-3/2})$. 
Assuming for the moment that this 
is valid, it is inferred that the only $O(1)$ contribution to $\int_{\partial \hat{D}_3} k^j E(t,k) dk$, $j = 1,2$, derives from the last line of (\ref{Eexpression}), 
which is to say, 
\begin{align*}
\int_{\partial \hat{D}_3} k^j E(t,k) dk
= & \, \re(K^2) \int_{\partial \hat{D}_3} \frac{k^j f'(k)}{f(k)} 
 \bigg\{ \frac{e^{i\omega t}}{4iK(f(k) + i\omega)} 
- \frac{e^{-i\omega t}}{4i\bar{K}(f(k) -i\omega)} \bigg\} dk
	\\
& + O\big(t^{-\frac{3}{2}}\big), \qquad t \to \infty, \quad j = 1,2.
\end{align*}
Computing the integral using the residue theorem,  the asymptotic relations (\ref{intkE}) emerge.  

It remains to show that any term in (\ref{Eexpression}) involving $e^{-f(k) t}$, $e^{-f(k') t}$, or $e^{-f(k'') t}$ yields a contribution of order $O\big(t^{-\frac{3}{2}}\big)$. This follows from steepest descent considerations. In Appendix \ref{QRapp}, the details of this calculation  are provided for the case of the triple integral
\begin{align}\nonumber
Q(t) := & \int_{\partial \hat{D}_3} dk \frac{k^2f'(k)}{f(k)} 
 \int_{\partial \hat{D}_3}dk' \frac{k'^2 f'(k')}{f(k') + i\omega}
	\\ \label{Qdef}
&\times  \int_{\partial \hat{D}_3}dk '' \frac{k''^2f'(k'')}{f(k'')(f(k'') + i\omega)}\frac{e^{-f(k'')t -f(k')t} - e^{-f(k)t}}{f(k) -f(k') - f(k'')}
\end{align}
and the double integral
\begin{align}\label{Rdef}
R(t) := \int_{\partial \hat{D}_3} dk \frac{k^2f'(k)}{f(k)} \int_{\partial \hat{D}_3}dk' \frac{k'^2 f'(k')}{f(k') + i\omega}
\frac{e^{i\omega t-f(k')t} - e^{-f(k)t}}{f(k) - f(k') + i\omega}.
\end{align}
The other terms can be treated similarly. 
The proof of Lemma \ref{Elemma} is complete.   
\proofend

Using equations (\ref{intkA1}), (\ref{intk2A1}), (\ref{intkjA2}), and (\ref{intkE}) in the definitions (\ref{F1def}) and (\ref{F2def}) of $F_1(t,k)$ and  $F_2(t,k)$, 
respectively, there appears  
\begin{align}
F_j(t,k) = O\big(t^{-\frac{3}{2}}\big), \qquad t \to \infty, \quad j = 1,2,
\end{align}
which completes the proof of Theorem \ref{mainth}.  \proofend

\section{Conclusion}\nequation

Two questions about solutions of a natural initial-boundary-value problem for the Korteweg-de Vries equation have been addressed here.  Both these 
questions arise naturally from observed experimental data obtained in water tank experiments.  Overall, these results, which are asymptotic, but exact 
in their large-time structure, raise a cautionary note.  While the positive result of asymptotic periodicity corresponds well to what is seen in 
experiments, the lack of asymptotically conserved mass is troubling.  Admittedly, this does not occur at first order, $O(\epsilon)$ in the notation
in force here, but rather at the second order $O(\epsilon^2)$.   As the Korteweg-de Vries model is only an accurate approximation on the 
so-called Boussinesq time
scale of order $O(1/\epsilon)$, the fact that mass is not conserved at the higher order does not make the initial-boundary-value problem 
considered here necessarily suspect.   However, it does reinforce the view that the model should not be pushed beyond the Boussinesq 
time scale.  If a unidirectional initial-boundary-value problem valid on a longer time scale is needed, a higher-order correct model should
be employed.  A recent example of such a model is provided in \cite{BCPS} (and see the references therein to other, related models)
 and a relevant initial-boundary-value problem was put forward and
analysed in \cite{chen}.  However, this latter problem features a piece of boundary data that might be hard to obtain in a laboratory setting.
   This issue deserves 
further study.

\appendix
\section{Steepest descent lemma} \label{steepestapp}
\renewcommand{\theequation}{A.\arabic{equation}}\nequation
In this appendix, the method of steepest descent is used  to determine the large $t$ behavior of certain integrals involving the exponential $\exp(-f(k) t)$. 

Let $a = 1/(2\sqrt{3})$ as before. The function $f'(k) = 2i(12k^2 - 1)$ vanishes at the two critical points $\pm a$. Let $\Gamma$ 
be the steepest descent contour shown in Figure \ref{Gamma.pdf}.  The contour  $\Gamma$ is characterized by the condition that $\im f(k) = \im f(a)$ on the part of $\Gamma$ passing through $k = a$, while $\im f(k) = \im f(-a)$ on the part of $\Gamma$ passing through $k = -a$.
For $j = 1,\dots, 4$,  let $\Gamma_j$ denote the part of $\Gamma$ that lies in the $j$'th quadrant. Then $\re f(k)$ is strictly increasing from $0$ to $+\infty$ as $k$ moves away from $\pm a$ towards $\infty$ along any of the $\Gamma_j$'s.

\begin{figure}
\begin{center}
\begin{overpic}[width=.5\textwidth]{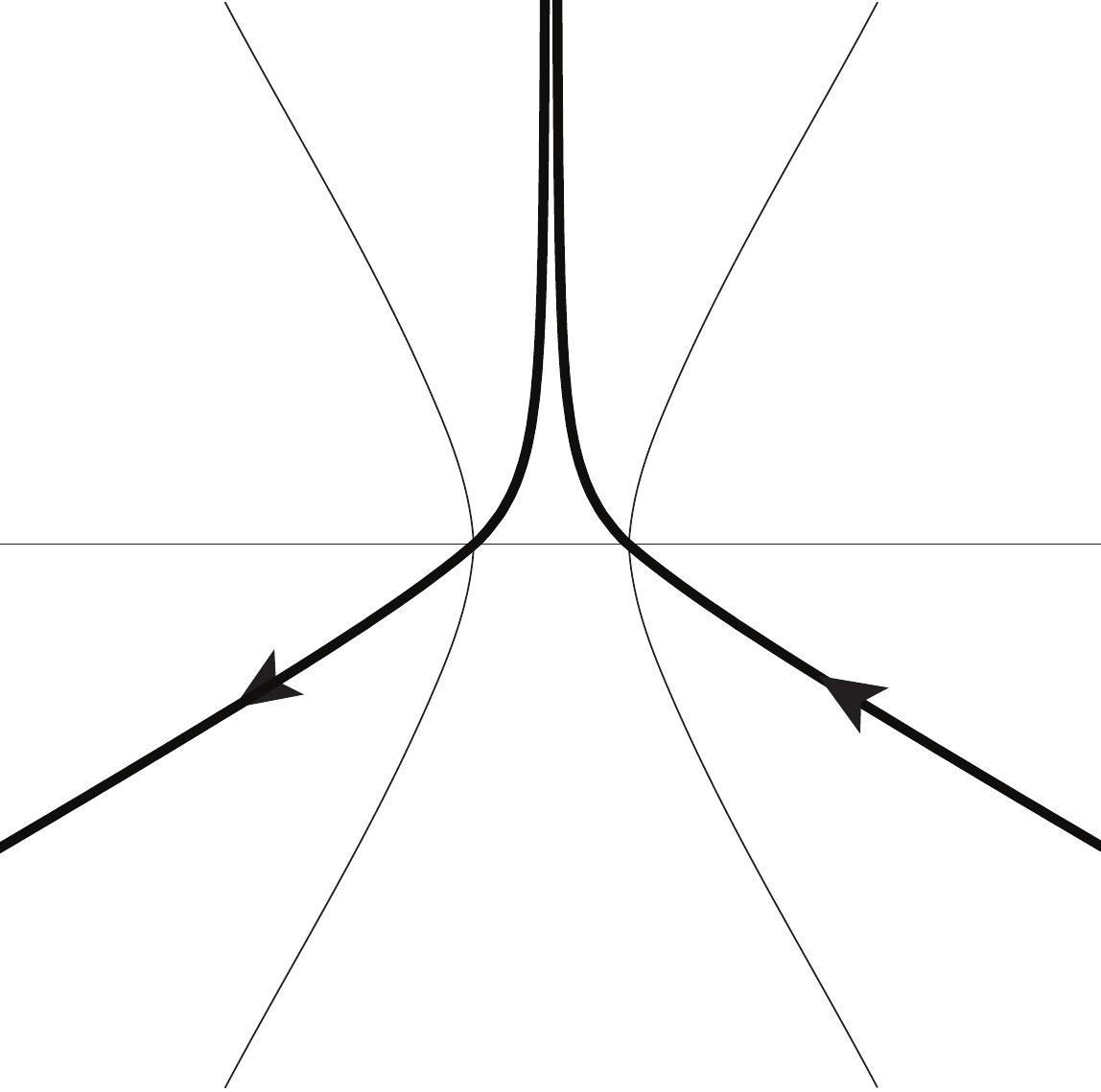}
      \put(80,39){$\Gamma$}
      \put(59,51.5){$a$}
      \put(34,51.5){$-a$}
      \end{overpic}
     \begin{figuretext}\label{Gamma.pdf}
       The steepest descent contour $\Gamma$ passing through the critical points $\pm a$.
     \end{figuretext}
     \end{center}
\end{figure}

\begin{proposition}\label{steepprop}
Let $q(k)$   be a function which is analytic in a neighborhood of $\Gamma$. Suppose that $q(k)$ grows at most algebraically as $k \to \infty$. It follows that 
\begin{align}\label{Gammaint}
\int_{\Gamma} q(k) f'(k) e^{-f(k) t} dk = O\big(t^{-\frac{3}{2}}\big), \qquad t \to \infty.
\end{align}
\end{proposition}
\proofbegin
Write the left-hand side of (\ref{Gammaint}) as the sum $\sum_{j=1}^4 I_j(t)$, where $I_j(t)$ denotes the contribution from $\Gamma_j$, {\it viz.} 
$$I_j(t) :=  e^{-f(a)t} \int_{\Gamma_j} q(k) f'(k) e^{-[f(k) - f(a)] t} dk, \qquad j = 1,\dots, 4.$$

Consider $I_1(t)$. 
For $k \in \Gamma_1$, let $l = f(k) -f(a)$. The  definition of the steepest descent contour implies that the mapping  $l \mapsto k(l)$ is a diffeomorphism 
from $[0, \infty)$  onto $ \Gamma_1$. Thus, a change of variables yields
\begin{align}\label{Isteepest}
I_1(t) = e^{-f(a)t} \int_0^{\infty} q(k(l)) e^{-l t} dl.
\end{align}
For any $\epsilon > 0$, the assumption that $q(k)$ grows at most algebraically as $k \to \infty$ implies that the integral
$$\bigg|\int_\epsilon^{\infty} q(k(l)) e^{-l t} dl\bigg| \leq e^{-\frac{\epsilon t}{2}} \int_\epsilon^{\infty} |q(k(l))| e^{-\frac{l t}{2}} dl \leq C e^{-\frac{\epsilon t}{2}}$$
is exponentially small. 
On the other hand, for $k$ near $a$ we have
$$k(l) = a + \frac{e^{\frac{3i\pi}{4}}}{2\cdot 3^{\frac{1}{4}}} \sqrt{l}
 + O(l), \qquad l \to 0.$$
Thus, if $q$ has the expansion
$$q(k) = \sum_{n=0}^\infty q_n (k - a)^n, \qquad k \to a,$$
then
$$q(k(l)) = q_0 + \frac{e^{\frac{3i\pi}{4}} q_1}{2\cdot 3^{\frac{1}{4}}} \sqrt{l}
+ O(l), \qquad l \to 0.$$
Substituting this into (\ref{Isteepest}) and evaluating the integrals with respect to $l$ leads to
\begin{align}\label{I1}
I_1(t) = e^{-f(a)t} \bigg\{\frac{q_0}{t}
+ \frac{e^{\frac{3i\pi}{4}} q_1}{2\cdot 3^{\frac{1}{4}}} \frac{\sqrt{\pi}}{2t^{\frac{3}{2}}}
+ O\big(t^{-2}\big) \bigg\}, \qquad t \to  \infty.
\end{align}
The last step can be made rigorous using standard arguments from the steepest descent method (see {\it e.g.} \cite{O1974}).  
A similar argument applied to  $\Gamma_4$ provides the asymptotic relation
\begin{align}\label{I4}
I_4(t) = e^{-f(a)t} \bigg\{-\frac{q_0}{t}
+ \frac{e^{\frac{3i\pi}{4}} q_1}{2\cdot 3^{\frac{1}{4}}} \frac{\sqrt{\pi}}{2t^{\frac{3}{2}}}
+  O\big(t^{-2}\big) \bigg\}, \qquad t \to \infty.
\end{align}
Equations (\ref{I1}) and (\ref{I4}) imply that $I_1(t)  + I_4(t) = O(t^{-3/2})$. 

Analogous computations give $I_2(t)  + I_3(t) = O(t^{-3/2})$, thereby establishing Proposition  (\ref{Gammaint}).
\proofend

\begin{remark}\upshape
The proof of Proposition \ref{steepprop} can   be extended to give the expansion of the integral in formula (\ref{Gammaint}) to all orders in $t$. In principle,
a tedious computation can then provide the asymptotic expansions in (\ref{gformulas}) to all orders in $t$. 
\end{remark}

\section{Asymptotics of $Q(t)$ and $R(t)$} \label{QRapp}
\renewcommand{\theequation}{B.\arabic{equation}}\nequation
\vspace{.4cm} 
Here, a proof is offered  that $Q(t)$ and $R(t)$ defined in (\ref{Qdef})-(\ref{Rdef}) are also of order $O(t^{-3/2})$ as $t \to \infty$.

{\bf Asymptotics of $Q(t)$}
To prove that $Q(t) = O(t^{-3/2})$, note that the integrand in (\ref{Qdef}) has removable singularities at the points where $f(k) - f(k') - f(k'') =0$. Moreover, the integrand is non-singular for $k$ in the set $\{0, K, -\bar{K}, L, -\bar{L}\}$, so the contour   $\partial \hat{D}_3$ in the $k$-integral  may be replaced  with $\partial D_3$. Deforming the contour of integration in the $k''$-integral from $\partial \hat{D}_3$ to $\Gamma$ and interchanging the order of the $k'$ and $k''$ integrals, it is found that 
\begin{align*}
Q(t) = & \int_{\partial D_3} dk \frac{k^2f'(k)}{f(k)} 
\int_{\Gamma}dk '' \frac{k''^2f'(k'')}{f(k'')(f(k'') + i\omega)}
 \int_{\partial \hat{D}_3}dk' \frac{k'^2 f'(k')}{f(k') + i\omega} 
	\\
&\times\frac{e^{-f(k'')t -f(k')t} - e^{-f(k)t}}{f(k) -f(k') - f(k'')}.
\end{align*}
Next, deform the contour in the $k$-integral so that it passes to the left (i.e. the indentation lies in $D_3$) of $k = 0$ as well as to the left of the solutions in $\partial D_3$ of the equations $f(k) = 2f(a)$ and $f(k) = -2f(a)$ (see Figure \ref{D3checked.pdf}). Denote this deformed contour by $\partial \check{D}_3$.
\begin{figure}
\begin{center}
      \begin{overpic}[width=.5\textwidth]{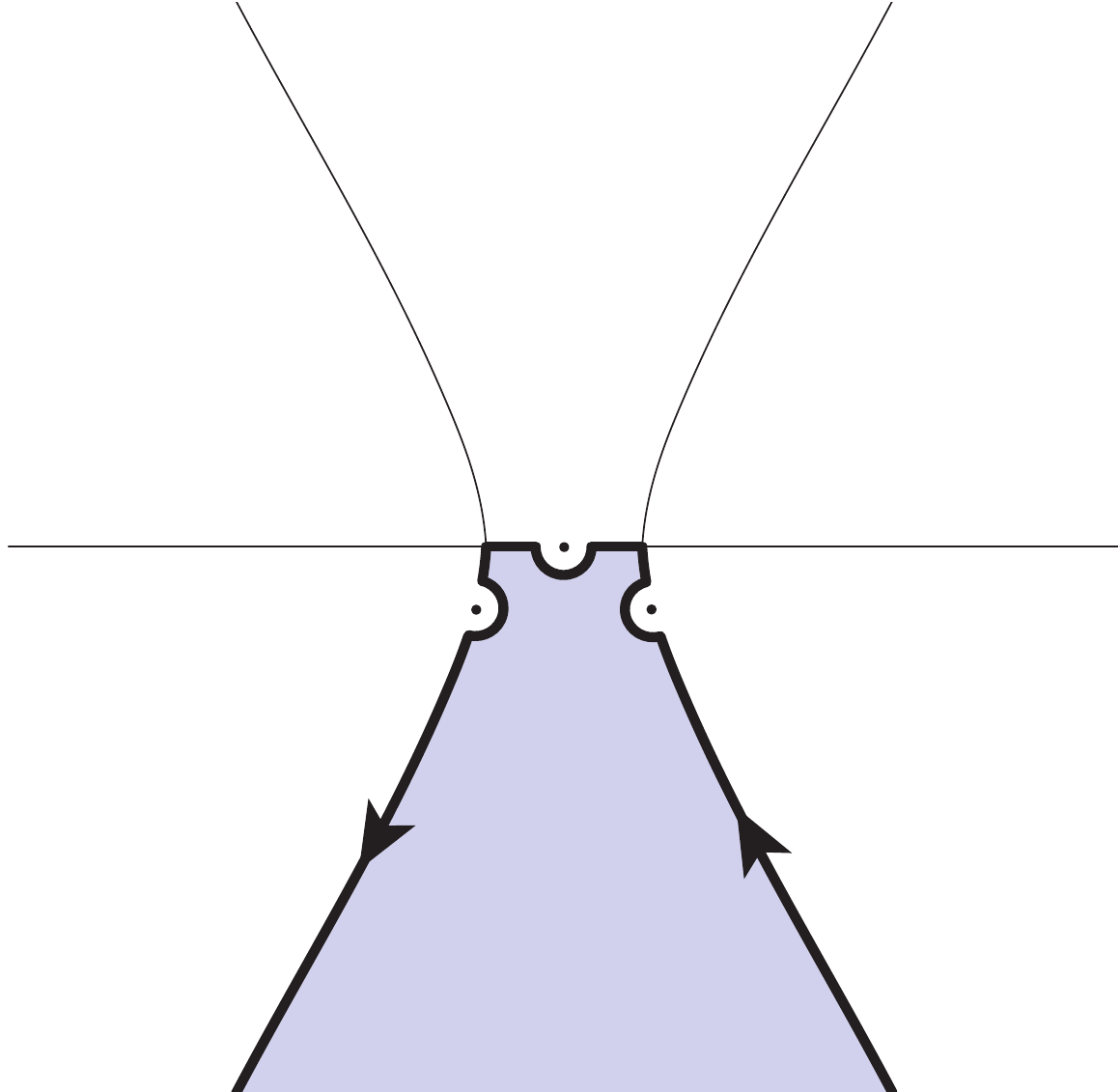}
      \put(71,24){$\partial \check{D}_3$}
      \end{overpic}
     \begin{figuretext}\label{D3checked.pdf}
       The contour $\partial \check{D}_3$.
     \end{figuretext}
     \end{center}
\end{figure}
For each pair $(k,k'') \in \partial \check{D}_3 \times \Gamma$, one observes that  $\re(f(k) - f(k'')) \leq 0$.  In consequence, 
 there exists a unique point $K_r(k,k'') \in \bar{D}_3$ such that
\begin{align}\label{fff}
f(k) - f(k'') = f(K_r(k,k'')).
\end{align}
Let $\partial \hat{D}_3(k,k'')$ denote the contour $\partial \hat{D}_3$ with a small indentation added so that it passes to the right of $K_r(k,k'')$ whenever $K_r(k,k'') \in \partial D_3$. The following claim implies that this indentation can be chosen in such a way that $\partial \hat{D}_3(k,k'') \subset \bar{D}_2 \cup \bar{D}_4$ for all $k, k''$.

\begin{lemma}  There exists an $\epsilon > 0$ such that 
$K_r(k,k'')$ stays an $\epsilon$-distance away from the critical points $\pm a$ for $(k,k'') \in \partial \check{D}_3 \times \Gamma$.
\end{lemma}
\proofbegin
Let $\Gamma_+ = \Gamma_1 \cup \Gamma_4$ denote the portion  of $\Gamma$ that passes through $a$. It will be shown 
 that $K_r(k,k'')$ stays well  away from $\pm a$ for $k \in \partial \check{D}_3$ and $k'' \in \Gamma_+$.  
 The case when $k'' \in \Gamma_2 \cup \Gamma_3$ can be handled in a similar way. 

In view of (\ref{fff}), $K_r(k,k'')$ can approach $\pm a$ only if $f(k) - f(k'')$ is close to $f(\pm a) = \mp 2i/(3\sqrt{3})$. Let $\epsilon > 0$ be such that $\dist(f(k), \{\pm 2f(a), 0\}) > 2\epsilon$ for $k \in \partial \check{D}_3$. There exists a $\delta > 0$ such that $|f(k'') - f(a)| < \epsilon$ whenever $k'' \in \Gamma_+$ satisfies $|k'' - a| < \delta$. 
Thus, if $k'' \in \Gamma_+$ is such that $|k'' - a| < \delta$ and $k \in \partial \check{D}_3$, then by the triangle inequality,
\begin{align*}
&  |f(k) - f(k'') - f(\pm a)| \geq |f(k) - f(a) - f(\pm a)| - |f(k'') - f(a)|
> 2\epsilon - \epsilon = \epsilon,
\end{align*}
so that $K_r(k,k'')$ stays away from $\pm a$ in this case.
On the other hand, since $\re f(k'') \geq 0$ for $k'' \in \Gamma_+$ and $\re f(k'')$ is small only for $k''$ near $a$, there exists a $c > 0$ such that $\re{f(k'')} > c$ whenever $|k'' - a| \geq \delta$. Also, $\re f(k) \leq 0$ for $k \in \partial \check{D}_3$.
Hence, if $k'' \in \Gamma_+$ is such that $|k'' - a| \geq \delta$ and $k \in \partial \check{D}_3$, then
$$|f(k) - f(k'') - f(\pm a)| \geq |\re(f(k) - f(k'') - f(\pm a))| = \re{f(k'')} - \re{f(k)} > c,$$
showing that $K_r(k,k'')$ also stays away from $\pm a$  in this case.
\proofend 

Since the contour $\partial \hat{D}_3(k,k'')$ avoids the point $K_r(k, k'')$, it is the case that 
$f(k) - f(k') -f(k'') \neq 0$ for all $(k,k'') \in \partial \check{D}_3 \times \Gamma$ and $k' \in \partial \hat{D}_3(k,k'')$. Hence, the integral $Q$ may be split as follows:
\begin{align*}
Q(t) = & \int_{\partial \check{D}_3} dk \frac{k^2f'(k)}{f(k)} 
 \int_{\Gamma}dk '' \frac{k''^2f'(k'')}{f(k'')(f(k'') + i\omega)} \int_{\partial \hat{D}_3(k,k'')} dk' \frac{k'^2 f'(k')}{f(k') + i\omega} 
 	\\
&\times \frac{e^{-f(k'')t -f(k')t}}{f(k) -f(k') - f(k'')}
 	\\
& - \int_{\partial \check{D}_3} dk \frac{k^2f'(k)}{f(k)} 
 \int_{\Gamma} dk '' \frac{k''^2f'(k'')}{f(k'')(f(k'') + i\omega)} \int_{\partial \hat{D}_3(k,k'')}dk' \frac{k'^2 f'(k')}{f(k') + i\omega}
  	\\
&\times\frac{e^{-f(k)t}}{f(k) -f(k') - f(k'')} =: Q_1(t) - Q_2(t).
\end{align*}

To compute $Q_1(t)$,  first use Jordan's lemma to deform the contour $\partial \hat{D}_3(k,k'')$ to $\Gamma$ and then deform the contour $\partial \check{D}_3$ in the $k$-integral to infinity.
Since $k', k'' \in \Gamma$ implies $f(k'), f(k'') \in \pm \frac{2i}{3\sqrt{3}} + \R_{\geq 0}$, it follows that $$|f(k) - f(k') - f(k'')| \geq \text{dist}\bigg(f(k), \bigg\{\pm \frac{4i}{3\sqrt{3}}, 0\bigg\}\bigg), \qquad k \in \bar{D}_3, \;\; k', k'' \in \Gamma.$$
This implies that there exists a $c > 0$ such that $|f(k) - f(k') - f(k'')| \geq c |k|^3$ for all large $k \in \bar{D}_3$. In particular, the integrand is $O(k^{-2})$ as $k \to \infty$ in $\bar{D}_3$. Thus, deforming the contour in the $k$-integral to infinity, observing that $f(k) -f(k') - f(k'')$ has strictly negative real part and hence is nonzero throughout this deformation, it follows  that $Q_1(t) = 0$. 

To compute $Q_2(t)$, remark that by deforming the contour $\partial \hat{D}_3(k,k'')$ to infinity and using Cauchy's theorem, the formula 
\begin{align*}
Q_2(t) =  &\; 2\pi i \int_{\partial \check{D}_3} dk \frac{k^2f'(k)e^{-f(k)t}}{f(k)} 
 \int_{\Gamma}dk '' \frac{k''^2f'(k'')}{f(k'')(f(k'') + i\omega)}
	\\
& \times  \bigg( \underset{k' = K}{\res} +
  \underset{k' = K_r(k,k'')}{\res}\bigg)
 \frac{k'^2 f'(k')}{f(k') + i\omega} \frac{1}{f(k) -f(k') - f(k'')} 	
 	\\
=  &\; 2\pi i \int_{\partial \check{D}_3} dk \frac{k^2f'(k)e^{-f(k)t}}{f(k)} 
 \int_{\Gamma}dk '' \frac{k''^2f'(k'')}{f(k'')(f(k'') + i\omega)}  \frac{K^2 - K_r^2(k,k'')}{f(k) - f(k'') + i\omega}
\end{align*}
emerges.  The function $f(k) - f(k'') + i\omega = 0$ vanishes at $k''= K_r(k,K)$, but this singularity is removable since $K_r(k, K_r(k,K)) = K$. 
In consequence, deforming the contour $\Gamma$ to infinity and using Cauchy's theorem yields
\begin{align*}
Q_2(t) =  &\; (2\pi i)^2 \int_{\partial \check{D}_3} dk \frac{k^2f'(k)e^{-f(k)t}}{f(k)} 
 \underset{k'' = K}{\res}\bigg( \frac{k''^2f'(k'')}{f(k'')(f(k'') + i\omega)}  \frac{K^2 - K_r^2(k,k'')}{f(k) - f(k'') + i\omega}\bigg) 	\\
=  &\; (2\pi i)^2 \int_{\partial \check{D}_3} dk \frac{k^2f'(k)e^{-f(k)t}}{f(k)} 
    \frac{K^2}{-i\omega}  \frac{K^2 - K_r^2(k,K)}{f(k) + 2i\omega}.	
\end{align*}
Since $K_r(L,K) = K$, the singularity at $k = L$ is removable.
Deforming the contour to $\Gamma$ and applying Proposition \ref{steepprop} yields $Q_2(t) = O(t^{-3/2})$, thereby showing that $Q(t) = O(t^{-3/2})$.  

\vspace{.4cm}

{\bf Asymptotics of $R(t)$}
To show that $R(t) = O(t^{-3/2})$, notice that the integrand in (\ref{Rdef}) has removable singularities at the points where $f(k) - f(k') + i\omega =0$.
To proceed, first deform the contour in the $k'$-integral from $\partial \hat{D}_3$ to $\Gamma$ and then  
 deform the contour in the $k$-integral from $\partial \hat{D}_3$ to $\partial \check{D}_3$.  However, contrary to the earlier notation,  
 the contour $\partial \check{D}_3$  is now obtained from $\partial D_3$ by inserting indentations so that it passes to the left of the origin and of the points $K_r(K, \pm a) \in \partial D_3$ at which $f(k) \pm f(a) + i\omega = 0$ (see again Figure  \ref{D3checked.pdf}). 
Since $\re f(k) \leq 0$ for $k \in \partial \check{D}_3$ and $\re f(k') \geq 0$ for $k' \in \Gamma$ with equality only at $k' = \pm a$, $f(k) - f(k') + i\omega$ can vanish only when $k' = \pm a$ and $k = K_r(K, \mp a)$. In particular, $f(k) - f(k') + i\omega \neq 0$ for $k \in \partial \check{D}_3$ and $k' \in \Gamma$. Thus the integral can be split, {\it viz.}
\begin{align*}
R(t) = & \;e^{i\omega t}\int_{\partial \check{D}_3} dk \frac{k^2f'(k)}{f(k)} \int_{\Gamma}dk' \frac{k'^2 f'(k')}{f(k') + i\omega}
\frac{e^{-f(k')t}}{f(k) - f(k') + i\omega} 
	\\
& - \int_{\partial \check{D}_3} dk \frac{k^2f'(k)}{f(k)} \int_{\Gamma}dk' \frac{k'^2 f'(k')}{f(k') + i\omega}
\frac{e^{-f(k)t}}{f(k) - f(k') + i\omega} =: R_1(t) - R_2(t).
\end{align*}
 
To compute $R_1(t)$, observe that $k' \in \Gamma$ implies $f(k') \in \pm \frac{2i}{3\sqrt{3}} + \R_{\geq 0}$. It thus transpire that 
$$|f(k) - f(k') + i\omega| \geq \text{dist}\bigg(f(k), i\omega \pm \frac{2i}{3\sqrt{3}}\bigg), \qquad k \in \bar{D}_3, \;\; k' \in \Gamma.$$
Thus, the integrand is $O(k^{-2})$ as $k \to \infty$   in $\bar{D}_3$. Deforming the contour in the $k$-integral to infinity in $D_3$, it follows that $R_1(t) = 0$. 

For the computation of  $R_2(t)$, remark that 
\begin{align*}
R_2(t) = & \;2\pi i \int_{\partial \check{D}_3} dk \frac{k^2f'(k)e^{-f(k)t}}{f(k)} \bigg(\underset{k' = K}{\res} + \underset{k' = K_r(k,K)}{\res}\bigg) \frac{k'^2 f'(k')}{f(k') + i\omega} \frac{1}{f(k) - f(k') + i\omega} 
	\\
= & \; 2\pi i \int_{\partial \check{D}_3} dk \frac{k^2f'(k)e^{-f(k)t}}{f(k)}
\frac{K^2 - K_r^2(k,K)}{f(k) + 2i\omega}.	
\end{align*}
Since $K_r(L,K) = K$, the singularity at $k = L$ is removable. Deforming the contour to $\Gamma$ and applying Proposition \ref{steepprop}, 
it is concluded that 
 $R_2(t) = O(t^{-3/2})$, thereby establishing  that $R(t) = O(t^{-3/2})$.

\section{An alternative perturbative approach} \label{perturbapp}
\renewcommand{\theequation}{C.\arabic{equation}}\nequation
In Theorem \ref{mainth} it was shown that if $g_0(t) = \epsilon \sin{\omega t}$, then $g_1$ and $g_2$ are asymptotically periodic as $t \to \infty$, 
at least  to second order in perturbation theory. We also computed the large $t$ asymptotics and gave rigorous error estimates for $g_1$ and $g_2$ to the same order.

If one {\it assumes} that $g_1$ and $g_2$ are asymptotically periodic as $t \to \infty$ with period $2\pi/\omega$, and if one does not worry about precise error
 estimates, the coefficients in (\ref{gformulas}) can be determined 
  directly using an alternative perturbative approach.  This idea was
   first implemented for the nonlinear Schr\"odinger equation in \cite{FLtperiodic}.

The KdV equation (\ref{kdv}) admits the Lax pair
\begin{align}\label{lax}
\begin{cases}
\varphi_x + ik\sigma_3 \varphi = V_1 \varphi,
	\\
\varphi_t + i(4k^3 - k)\sigma_3 \varphi = V_2 \varphi,
\end{cases}
\end{align}
where $\varphi(x,t,k)$ is a vector-valued eigenfunction, $k \in \C$ is the spectral parameter, $\{V_j(x,t,k)\}_1^2$ are defined by
\begin{align*}
& V_1 = \frac{u}{2k} (\sigma_2 - i \sigma_3), \qquad V_2 = 2ku\sigma_2 + u_x \sigma_1 + \frac{2u^2 + u + u_{xx}}{2k}(i\sigma_3 - \sigma_2),
\end{align*}
and $\{\sigma_j\}_1^3$ denote the standard Pauli matrices. Letting $\varphi_1$ and $\varphi_2$ denote the first and second entries of $\varphi$ respectively, the $t$-part of (\ref{lax}) can be written as
\begin{align*}
 & \varphi_{1t} + \frac{i}{2k}(8k^4 - 2k^2 - u - 2u^2 - u_{xx}) \varphi_1 
  + \Big[2ik u - u_x - \frac{i}{2k}(u + 2u^2 + u_{xx})\Big]\varphi_2 = 0,
  	\\
& \varphi_{2t} - \frac{i}{2k}(8k^4 - 2k^2 - u - 2u^2 - u_{xx})\varphi_2
- \Big[2iku + u_x - \frac{i}{2k}(u + 2u^2 + u_{xx})\Big] \varphi_1 = 0.	
\end{align*}
This in turn  implies that the quotient $p = \varphi_1/\varphi_2$ satisfies the Ricatti equation
\begin{align}\nonumber
& i p_t + \frac{1}{2k}((1- 4k^2) u + 2u^2 + 2ik u_x + u_{xx}) p^2
+ \frac{1}{k}(2k(k -4k^3) + u + 2u^2 + u_{xx}) p
	\\ \label{ratioeq}
& \hspace{3.35cm}  + \frac{1}{2k}((1 - 4k^2) u + 2u^2 - 2ik u_x + u_{xx}) = 0.
\end{align}

{\it Assume} that the functions $\{g_j(t)\}_0^2$ and $p(t,k)$ are asymptotically time-periodic of period $t_p = 2\pi /\omega$ as $t \to \infty$.  
The functions to which they are asymptotic then have Fourier series representations 
\begin{align} \nonumber
 & g_0(t) \sim \sum_{n=-\infty}^\infty a_n e^{in\omega t} , \qquad 
  g_1(t) \sim \sum_{n=-\infty}^\infty b_n e^{in\omega t}, 
  	\\\label{g1dns}
&  g_2(t) \sim \sum_{n=-\infty}^\infty c_n e^{in\omega t}, \qquad
p(t,k) \sim \sum_{n=-\infty}^\infty d_n(k) e^{in\omega t}, \qquad t \to \infty.
\end{align}   
Substituting these representations into equation (\ref{ratioeq}) evaluated at $x = 0$ leads directly to the equation
\begin{align*}
& \sum_{n=-\infty}^\infty 
\biggl\{- n\omega d_n(k) 
+ \frac{1 - 4k^2}{2k} \sum_{l,m=-\infty}^\infty a_l  d_m(k)d_{n-l-m}(k)
	\\
& \hspace{1cm}  +  \frac{1}{k}\sum_{j,l,m=-\infty}^\infty a_j a_l d_m(k)d_{n-j-l-m}(k) 
   + i \sum_{l,m=-\infty}^\infty b_l d_m(k)d_{n-l-m}(k) 
	\\
& \hspace{1cm}   +  \frac{1}{2k} \sum_{l,m=-\infty}^\infty c_l d_m(k)d_{n-l-m}(k) 
  + 2(k - 4k^3)  d_n(k)  
+ \frac{1}{k}\sum_{m=-\infty}^\infty a_m d_{n-m}(k) 
	\\
& \hspace{1cm}   + \frac{2}{k}\sum_{l,m=-\infty}^\infty a_l a_m d_{n-l-m}(k) 
+ \frac{1}{k} \sum_{m=-\infty}^\infty c_m d_{n-m}(k) 
	\\ 
&  \hspace{1cm}   + \frac{1 - 4k^2}{2k} a_n  
+ \frac{1}{k}\sum_{m=-\infty}^\infty a_m a_{n-m}   
 - i  b_n   + \frac{1}{2k} c_n \biggr\} e^{in\omega t}  = 0.
\end{align*}
This in turn yields the infinite hierarchy of equations
\begin{align}\nonumber
 d_n(k)
= &\; \frac{i}{f(k) + in\omega}\biggl\{
\frac{1 - 4k^2}{2k} \sum_{l,m=-\infty}^\infty a_l  d_m(k)d_{n-l-m}(k)
	\\\nonumber
&  +  \frac{1}{k}\sum_{j,l,m=-\infty}^\infty a_j a_l d_m(k)d_{n-j-l-m}(k) 
   + i \sum_{l,m=-\infty}^\infty b_l d_m(k)d_{n-l-m}(k) 
	\\\nonumber
& +  \frac{1}{2k} \sum_{l,m=-\infty}^\infty c_l d_m(k)d_{n-l-m}(k)   
+ \frac{1}{k}\sum_{m=-\infty}^\infty a_m d_{n-m}(k) 
	\\ \nonumber
& + \frac{2}{k}\sum_{l,m=-\infty}^\infty a_l a_m d_{n-l-m}(k) 
+ \frac{1}{k} \sum_{m=-\infty}^\infty c_m d_{n-m}(k) 
	\\ \label{algebraicsystem}
& + \frac{1 - 4k^2}{2k} a_n  
+ \frac{1}{k}\sum_{m=-\infty}^\infty a_m a_{n-m}  
 - i  b_n   + \frac{1}{2k} c_n\biggr\}, \quad  n \in \Z.
\end{align}

\subsection{Perturbative solution of the algebraic system}
The algebraic system (\ref{algebraicsystem}) may be solved perturbatively 
if the Fourier coefficients $a_n$ associated with the Dirichlet data are known. 
Indeed, a perturbative analysis of (\ref{algebraicsystem}) yields expressions for the coefficients $d_n(k)$ in terms of the $a_n$, $b_n$, and $c_n$. The condition that $d_n(k)$ be non-singular in $\bar{D}_1$ then yields expressions for the Fourier coefficients $b_n$ and $c_n$ associated with the Neumann values in terms of the $a_n$. This provides a straightforward, constructive approach to the Dirichlet to Neumann map.

In more detail, first substitute the expansions
\begin{align*}
& a_n = \epsilon a_{1,n} + \epsilon^2 a_{2,n} + O(\epsilon^3), && \epsilon \to 0, \quad n \in \Z,
	\\
& b_n = \epsilon b_{1,n} + \epsilon^2 b_{2,n} + O(\epsilon^3), && \epsilon \to 0, \quad n \in \Z,
	\\
& c_n = \epsilon c_{1,n} + \epsilon^2 c_{2,n} + O(\epsilon^3), && \epsilon \to 0, \quad n \in \Z,	
	\\
& d_n(k) = \epsilon d_{1,n}(k) +  \epsilon^2 d_{2,n}(k) + O(\epsilon^3), && \epsilon \to 0, \quad n \in \Z,
\end{align*}
into (\ref{algebraicsystem}).  The terms of $O(\epsilon)$ lead to the formulas 
\begin{align}\label{orderepsilon}
d_{1,n}(k) = \frac{i}{f(k) + in\omega}\biggl\{
\frac{1 - 4k^2}{2k} a_{1,n} - i  b_{1,n} + \frac{1}{2k} c_{1,n}\biggr\}, \qquad n \in \Z.
\end{align}
For $n \in \Z$, let $k_1(n)$, $k_2(n)$, and $k_3(n)$ denote the three roots of $f(k) + in\omega = 0$, ordered so that 
$$k_1(n) \in \partial D_1', \qquad k_2(n) \in \partial D_1'', \qquad k_3(n) \in \partial D_3.$$ 
These roots satisfy the identities
\begin{align}\nonumber
& k_1(n) + k_2(n) + k_3(n) = 0, \qquad k_1(n)k_2(n) + k_1(n)k_3(n) + k_2(n)k_3(n) = -\frac{1}{4},
	\\ \label{krelations}
& k_1(n)k_2(n)k_3(n) = -\frac{n\omega}{8}, \qquad k_3(-n) = -\overline{k_3(n)}, \qquad n \in \Z.
\end{align}

For $d_{1,n}(k)$ to be nonsingular at $k_1(n)$ and $k_2(n)$, it is required that 
$$\frac{1 - 4k_j^2(n)}{2k_j(n)} a_{1,n} - i  b_{1,n} + \frac{1}{2k_j(n)} c_{1,n} = 0, \qquad n \in \Z, \quad j = 1,2.$$
Solving these equations for $\{b_{1,n}, c_{1,n}\}$ and using (\ref{krelations}), there appears the formulas
\begin{align}\label{bcsolvedfor}
& b_{1,n} = -2ik_3(n) a_{1,n},  \qquad c_{1,n} = -4k_3^2(n) a_{1,n}, \qquad n \in \Z.
\end{align}
The equations (\ref{orderepsilon}) thus lead to 
\begin{align}\label{d1nk}
d_{1,n}(k) = \frac{ia_{1,n}}{f(k) + in\omega}\biggl\{
\frac{1 - 4k^2}{2k} - 2k_3(n)  - \frac{2}{k} k_3^2(n) \biggr\}, \qquad n \in \Z.
\end{align}
Similarly, the terms of $O(\epsilon^2)$ give
\begin{align} \nonumber
d_{2,n}(k)
= &\; \frac{i}{f(k) + in\omega}\biggl\{
\frac{1}{k}\sum_{m=-\infty}^\infty a_{1,m} d_{1,n-m}(k) 
+ \frac{1}{k} \sum_{m=-\infty}^\infty c_{1,m} d_{1,n-m}(k) 
	\\ \label{orderepsilon2}
&+ \frac{1 - 4k^2}{2k} a_{2,n} + \frac{1}{k}\sum_{m=-\infty}^\infty a_{1,m} a_{1,n-m}   
 - i  b_{2,n}   + \frac{1}{2k} c_{2,n}\biggr\}, \qquad n \in \Z.
\end{align}
The condition that $d_{2,n}(k)$ should not have singularities at $k_1(n)$ and $k_2(n)$ determines the coefficients $b_{2,n}$ and $c_{2,n}$. 
This process can be continued indefinitely. Indeed, the terms of order $O(\epsilon^m)$ yield an equation of the form
\begin{align*}
d_{m,n}(k) = &\;\frac{i}{f(k) + in\omega}\biggl\{
F_{mn}(k)  + \frac{1 - 4k^2}{2k} a_{m,n} - i  b_{m,n}   + \frac{1}{2k} c_{m,n} \biggr\}, \qquad n \in \Z,
\end{align*}
where the function $F_{mn}(k)$ is given in terms of (known) lower order terms.
The condition that $d_{m,n}(k)$ should not have singularities at $k_1(n)$ and $k_2(n)$ then implies that 
\begin{align*}
& b_{m,n} = 2i a_{m,n}(k_1(n) + k_2(n))
- i\frac{ k_1(n) F_{mn}(k_1(n)) - k_2(n)F_{mn}(k_2(n))}{k_1(n) - k_2(n)},
	\\
& c_{m,n} = -a_{m,n}(1 + 4k_1(n) k_2(n))
+\frac{2 k_1(n) k_2(n)(F_{mn}(k_1(n)) - F_{mn}(k_2(n)))}{k_1(n) - k_2(n)}.
\end{align*}
Once $\{b_{m,n}\}_{n=-\infty}^\infty$, $\{c_{m,n}\}_{n=-\infty}^\infty$, $\{d_{m,n}(k)\}_{n=-\infty}^\infty$ are determined, one may proceed to the next order.

\begin{example}\upshape
Consider the example wherein
$$g_0(t) = \epsilon \sin \omega t$$
with $\omega \neq 0$ satisfying (\ref{omegaassumption}). Then $a_{m,n} = 0$ for all $m \geq 2$ and 
$$a_{1,n} = \begin{cases} 
\frac{1}{2i}, & n = 1, 
	\\
-\frac{1}{2i}, & n = -1, 
	\\
0, & \text{otherwise}.
\end{cases}$$
As in Theorem \ref{mainth},  write $K$ and $L$ for $k_3(1)$ and $k_3(2)$ respectively. Equations (\ref{bcsolvedfor}) and (\ref{d1nk}) provide 
the formulas
\begin{align}\nonumber
&b_{1,1} = \bar{b}_{1,-1} = -K, \qquad c_{1,1} = \bar{c}_{1,-1} = 2i K^2,
	\\\nonumber
& d_{1,1}(k) = \frac{1}{2(f(k) + i\omega)}\biggl\{
\frac{1 - 4k^2}{2k} - 2K  - \frac{2K^2}{k} \biggr\},
	\\\nonumber
& d_{1,-1}(k) = -\frac{1}{2(f(k) - i\omega)}\biggl\{
\frac{1 - 4k^2}{2k} + 2\bar{K}  - \frac{2\bar{K}^2}{k} \biggr\},
	\\\label{order1sol}
& b_{1,n} = c_{1,n} = d_{1,n} = 0, \qquad n \neq \pm 1,
\end{align}
whence 
\begin{align*}
& g_{11}(t) \sim - K  e^{i\omega t}  -\bar{K}  e^{-i\omega t}
\;\;  {\rm and} \;\;  g_{21}(t) \sim  2iK^2 e^{i\omega t}- 2i \bar{K}^2 e^{-i\omega t},   \;\; t \to \infty.
\end{align*}
Thus, apart from the error terms, the expressions in (\ref{g11end}) and (\ref{g21end}) have been recovered.

Similarly, equation (\ref{orderepsilon2}) yields
\begin{align*}
& d_{2,2}(k)
=  \; \frac{i}{f(k) + 2i\omega}\biggl\{
\frac{1}{k}a_{1,1} d_{1,1}(k) 
+ \frac{1}{k} c_{1,1} d_{1,1}(k) 
+ \frac{1}{k} a_{1,1}^2  
- i  b_{2,2}   + \frac{1}{2k} c_{2,2}\biggr\},
	\\
& d_{2,1}(k) = \; \frac{i}{f(k) + i\omega}\biggl\{- i  b_{2,1} + \frac{1}{2k} c_{2,1}\biggr\},
	\\
& d_{2,0}(k)
=  \;\frac{i}{kf(k)}\biggl\{
 (a_{1,1}+  c_{1,1} ) d_{1,-1}(k) 
+ (a_{1,-1} +  c_{1,-1}) d_{1,1}(k) 
	\\
& \hspace{3.95cm} + 2 a_{1,1} a_{1,-1} - i k b_{2,0} + \frac{1}{2} c_{2,0}\biggr\},
	\\
& d_{2,-1}(k)
= \; \frac{i}{f(k) - i\omega}\biggl\{ - i  b_{2,-1}   + \frac{1}{2k} c_{2,-1}\biggr\},
	\\
& d_{2,-2}(k)
= \; \frac{i}{f(k) - 2i\omega}\biggl\{
\frac{1}{k}a_{1,-1} d_{1,-1}(k) + \frac{1}{k} c_{1,-1} d_{1,-1}(k) 
+ \frac{1}{k} a_{1,-1}^2  
 - i  b_{2,-2}
 	\\
&   \hspace{4.15cm}  +\frac{1}{2k} c_{2,-2}\biggr\}.	
\end{align*}
For $d_{2,2}(k)$ to not have singularities at $k_1(2)$ and $k_2(2)$, it must be the case that 
$$\begin{cases}
(a_{1,1} + c_{1,1}) d_{1,1}(k_1(2)) + a_{1,1}^2 - i k_1(2) b_{2,2} + \frac{c_{2,2}}{2} = 0, 
	\\
(a_{1,1} + c_{1,1}) d_{1,1}(k_2(2)) + a_{1,1}^2 - i k_2(2) b_{2,2} + \frac{c_{2,2}}{2} = 0.
\end{cases}$$
Solving for $b_{2,2}$ and $c_{2,2}$ leads to 
\begin{align*}
& b_{2,2} = -i(a_{1,1} + c_{1,1})  \frac{d_{1,1}(k_1(2)) - d_{1,1}(k_2(2))}{k_1(2) - k_2(2)},
	\\
& c_{2,2} = -2 a_{1,1}^2 -2 (a_{1,1} +  c_{1,1}) \frac{k_1(2) d_{1,1}(k_2(2)) - k_2(2) d_{1,1}(k_1(2))}{k_1(2) - k_2(2)}.
\end{align*}
Substituting in  (\ref{order1sol}) and using the identities
$$k_1(2)k_2(2) = -\frac{\omega}{4L}, \qquad k_1(2) + k_2(2) = -L, \qquad 4L^2 - 4K^2 = -\frac{\omega}{L} + \frac{\omega}{2K},$$
we find after some algebraic manipulations that
$$b_{2,2} = \frac{1}{8i}\bigg(\frac{L}{K^2} - \frac{2}{K}\bigg),
\qquad c_{2,2} = 1 - \frac{L^2}{4K^2}.$$
For $d_{2,1}(k)$ to not have singularities at $k_1(1)$ and $k_2(1)$, it is required that 
$$\begin{cases}
-ib_{2,1} + \frac{c_{2,1}}{2k_1(1)} = 0, 
	\\
-ib_{2,1} + \frac{c_{2,1}}{2k_2(1)} = 0,
\end{cases}
$$
which implies that $b_{2,1} = c_{2,1} = 0$. 
Similar considerations yield
$$b_{2,0} = \frac{1}{2} \im\bigg(\frac{1}{K}\bigg),\qquad
c_{2,0} = \re\bigg(\frac{K}{\bar{K}}\bigg) -1, $$
and 
$$b_{2,-1} = c_{2,-1} = 0, \qquad b_{2,-2} = \bar{b}_{2,2}, \qquad
c_{2,-2} = \bar{c}_{2,2}.$$
\vspace{.3cm}
Apart from the error terms, this recovers the expansions in (\ref{g12end}) and (\ref{g22end}).
\end{example}

\bigskip
\noindent
{\bf Acknowledgement} {\it The gestation of this project took place when the authors were both participating in a conference 
held at the Schr\"odinger Institute in Vienna.  
The authors are especially grateful to the American Institute of Mathematics for providing an excellent environment for discussions during a weeklong 
National Science Foundation USA supported
 workshop on boundary-value problems for nonlinear, dispersive equations. JL thanks the University of Illinois at Chicago for support during a visit there and also acknowledges support from the European Research Council, Grant Agreement No. 682537, the Swedish Research 
 Council, Grant No. 2015-05430, the G\"oran Gustafsson Foundation, and the EPSRC, UK. JB is grateful to Southern University of Science 
and Technology in Shenzhen  for providing support for a visit there in the course of  which the manuscript was finalized.

\bibliographystyle{plain}
\bibliography{is}

\end{document}